\documentclass{amsart}

\usepackage{amsmath,amsfonts,amssymb,amsthm,graphicx,cite}
\usepackage{amscd}
\usepackage{mathrsfs}
\usepackage[all,cmtip]{xy}

\usepackage{tikz-cd}

\begin{document}

\newtheorem{theorem}{Theorem}[section]
\newtheorem{conjecture}[theorem]{Conjecture}
\newtheorem{corollary}[theorem]{Corollary}
\newtheorem{lemma}[theorem]{Lemma}
\newtheorem{proposition}[theorem]{Proposition}
\newtheorem{remark}[theorem]{Remark}

\newcommand{\Z}{{\mathbb Z}}
\newcommand{\N}{{\mathbb N}}
\newcommand{\C}{{\mathbb C}}
\newcommand{\R}{{\mathbb R}}
\newcommand{\Q}{{\mathbb Q}}

\newcommand{\cA}{\mathcal{A}}
\newcommand{\cB}{\mathcal{B}}
\newcommand{\cD}{\mathcal{D}}
\newcommand{\cH}{\mathcal{H}}
\newcommand{\cL}{\mathcal{L}}
\newcommand{\cQ}{\mathcal{Q}}

\newcommand{\fR}{\mathfrak{R}}

\newcommand{\re}{{\rm Re}\, }
\newcommand{\im}{{\rm Im}\, }
\newcommand{\Irrep}{{\rm Irrep}}
\newcommand{\Rad}{{\rm Ker}}
\newcommand{\Ker}{{\rm Ker}}

\def\beq#1#2\eeq{%
            \begin{equation}%
            \label{#1}%
                #2%
            \end{equation}%
        }

\title[Quasi--invariants and Lassalle--Nekrasov correspondence]{Quasi--invariant Hermite polynomials and Lassalle--Nekrasov correspondence}

\author{M.~V.~Feigin}
\email{misha.feigin@glasgow.ac.uk}
\address{School of Mathematics and Statistics, University of Glasgow, University Place, Glasgow G12 8QQ, UK  and  Faculty of Mechanics and Mathematics, Moscow State University, Moscow, Russia}
\author{M.~A.~Halln\"as}
\email{hallnas@chalmers.se}
\address{Department of Mathematical Sciences, Chalmers University of Technology and the University of Gothenburg, SE-412 96 Gothenburg, Sweden}
\author{A.~P.~Veselov}
\email{A.P.Veselov@lboro.ac.uk}
\address{Department of Mathematical Sciences, Loughborough University, Leicestershire LE11 3TU, UK; 
Faculty of Mechanics and Mathematics, Moscow State University and Steklov Mathematical Institute, Moscow, Russia
}

\date{\today}

\begin{abstract}
Lassalle and Nekrasov discovered in the 1990s a surprising correspondence between the rational Calogero--Moser system with a harmonic term and its trigonometric version. We present a conceptual explanation of this correspondence using the rational Cherednik algebra and establish its quasi-invariant extension. 

More specifically, we consider configurations $\mathcal A$ of real hyperplanes with multiplicities admitting the rational Baker--Akhiezer function and use this to introduce a new class of non-symmetric polynomials, which we call $\mathcal A$-Hermite polynomials. These polynomials form a linear basis in the space of $\mathcal A$-quasi-invariants, which is an eigenbasis for the corresponding generalised rational Calogero--Moser operator with harmonic term. In the case of the Coxeter configuration of type $A_N$ this leads to a quasi-invariant version of the Lassalle--Nekrasov correspondence and its higher order analogues.


\end{abstract}

\maketitle

\section{Introduction}

In 1971 Calogero \cite{Cal71} studied the quantum system
describing $N$ particles on the line pairwise
interacting  with the rational potential \beq{rat}
U_R(x_{1},\dots , x_{N})=\sum_{1\le i<j\le
N}\frac{\gamma}{(x_{i}-x_{j})^2} +\omega^2 x^2, \quad  x^2 = \sum_{i=1}^N x_i^2.
\eeq
Almost at the same time Sutherland \cite{Sut71} considered the quantum system of $N$ particles on a circle with trigonometric interaction
\beq{trig}
U_T(x_{1},\dots ,
x_{N})=\sum_{1\le i<j\le N}\frac{\gamma a^2}{\sin^2 a (x_{i}-x_{j})},\
\eeq 
which in the limit $a \to 0$ turns into the Calogero system without the harmonic term.

At the classical level both of these systems were studied in 1975 by Moser in a very influential paper \cite{Mos75}, who proved their integrability in Liouville sense by finding their Lax representations. Olshanetsky and Perelomov made another substantial contribution, in particular, by introducing the generalisations related to arbitrary root systems \cite{OP76A,OP83}.

Note that the dynamics in the rational and trigonometric cases are very different. In the rational case (\ref{rat}) the system is isochronous: all solutions are periodic with the same period $T=2\pi/\omega,$ while in the trigonometric case the motion is much more complicated \cite{OP76B}.

So it came as a surprise when in 1997 Nekrasov \cite{Nek97} discovered that the systems (both in the classical and the quantum case) are essentially equivalent.
More precisely, he showed that there is a symplectomorphism from the phase space of the rational system onto the open (positive) part of the phase space of the trigonometric system. 
Under this equivalence  integrals are mapped to  integrals with the Hamiltonian of the rational system being mapped to the momentum of the trigonometric system (which gives another proof of isochronicity of the rational system). 
To derive this map Nekrasov used the Hamiltonian reduction, following the ideas of the work by Kazhdan, Kostant and Sternberg \cite{KKS78}.

This explains an earlier construction of Lassalle \cite{Las91} of multivariable Hermite polynomials from Jack polynomials, which can be interpreted as a correspondence between the eigenfunctions of the two quantum systems with potentials \eqref{rat} and \eqref{trig} (see also Baker and Forrester's paper  \cite{BF97} which, in particular, contains further unpublished results due to Lassalle).
For these reasons we call this equivalence the {\it Lassalle--Nekrasov correspondence}.

One of the aims of this paper is to give a 
new explanation of the Lassalle--Nekrasov correspondence in the quantum case using the rational Cherednik algebra and extend it for the special parameter values $\gamma=2 m(m+1), \, m \in \mathbb Z_{\ge 0}$, from the symmetric to the much wider quasi-invariant setting. 

We start with a more general quantum system with the Hamiltonian
\beq{CMharm}
\mathcal H_\mathcal A =-\Delta + \sum_{\alpha\in\mathcal{A}}\frac{m_\alpha(m_\alpha+1)(\alpha,\alpha)}{(\alpha,x)^2}+ x^2,
\eeq
where $\mathcal A$ is any configuration of vectors in the Euclidean space  $\R^N$ with multiplicities $m_\alpha\in \Z_{\ge 0}$ admitting the so-called rational Baker--Akhiezer function $\phi(x,\lambda)$. Any Coxeter configuration belongs to this class, but there are also non-symmetric configurations (see \cite{CFV99} and the next section). 

Following \cite{CV90, CFV99, FV02} we define an algebra of {\it quasi-invariant}  polynomials $\mathcal Q_{\mathcal A}$ for any such configuration $\mathcal A$. It  consists of the  polynomials $q(x)$, $x \in \R^N$, which are  invariant up to order $2 m_\alpha$ with respect to orthogonal reflection about the hyperplane $(\alpha,x)=0$ for all $\alpha \in \mathcal A$, or equivalently, for each $\alpha\in \mathcal A$, the odd normal derivatives $\partial_\alpha^{2s-1} q(x) = (\alpha,\frac{\partial}{\partial x})^{2s-1} q(x)$ vanish on the hyperplane $(\alpha,x)=0$ for $s=1,2,\ldots,m_\alpha$.

Our starting observation is that after multiplication by a Gaussian factor the  rational BA function can be considered as the generating function of a new interesting class of quasi-invariant polynomials, which we call $\mathcal A$-{\it Hermite polynomials}. 
 In the symmetric case  multivariable  generalisations of Hermite polynomials were introduced and studied by Lassalle \cite{Las91}, with important subsequent contributions including works by Baker and Forrester \cite{BF97}, van Diejen \cite{vDie97} and R\"osler \cite{Ros98}.

 More precisely, choosing a homogeneous basis of quasi-invariants $q^i, \, i\in \mathbb Z_{\geq 0}$, we define the $\mathcal A$-Hermite polynomials $H_{q_i}(x)$ 
 by the generating function
\begin{equation}\label{AH}
\phi(x,\lambda)\exp(-\lambda^2/2)=\sum_{i=0}^\infty H_{q_i}(x)q^i(\lambda), \quad x, \lambda \in \R^N.
\end{equation}
The polynomial $H_{q_i}(x)$ is a non-homogeneous polynomial with the highest degree term $q_i(x)$, where collection of polynomials $q_j, \, j \in \mathbb Z_{\geq 0}$,  forms the dual basis in the space of quasi-invariants (see  details in Section~\ref{sec3}). 
Expansion \eqref{AH}  is a far-reaching generalisation of a well-known generating function expansion (see e.g. \cite{Sze39}) for the classical Hermite polynomials:
$$
\exp(\lambda x- \lambda^2/2) = \sum_{n=0}^\infty H_n(x) \frac{\lambda^n}{n!},
$$
corresponding to $m=0$ and $N=1.$

More generally, we define the {\it Hermitisation map} $\chi_H\colon \mathcal Q_{\mathcal A} \to \mathcal Q_{\mathcal A}$ sending $q$ to $H_q$, which can be given explicitly by 
\begin{equation}\label{simi}\nonumber
H_q(x)=(-1)^{\deg q}\exp(x^2/2)L_q\exp(-x^2/2)
\end{equation}
or, alternatively,   
\begin{equation}\label{simi2}\nonumber
H_q(x)=e^{-L/2} q(x),
\end{equation}
where $L=\Delta - \sum_{\alpha \in \mathcal A}\frac{2 m_\alpha}{(\alpha,x)}\partial_\alpha$ is a (conjugated) rational  Calogero--Moser (CM) Hamiltonian and $L_q$ is its quantum integral with highest order term $q(\partial)$ (see Section \ref{sec3}). 
These relations  generalise properties of the classical Hermite polynomials (see formulae (\ref{exp1}), (\ref{exp2}) below). Lassalle's generalised Hermite polynomials  \cite{Las91}
correspond  to the Coxeter configuration $\mathcal A$ of type $A_{N-1}$ 
and Jack symmetric polynomials $q$.

The corresponding functions $\Psi_q(x)=A_m(x)^{-1}H_q(x)e^{-x^2/2}$, where
$
A_m(x) = \prod_{\alpha\in\mathcal{A}}(x,\alpha)^{m_\alpha},
$
are (formal) eigenfunctions of the generalised CM Hamiltonian \eqref{CMharm} with harmonic term.


In Section \ref{sect5} we study in more detail the case of the Coxeter configuration $\mathcal A$ of type $A_{N-1}$ with all roots of multiplicity $m.$ 
In particular, we extend the Lassalle--Nekrasov correspondence to the quasi-invariants.
More precisely, we show that the map $\chi_H$ intertwines the action on $\cQ_\mathcal A$ of  quantum integrals $\cL_p, \, p \in \R[\xi_1,\ldots,\xi_N]^{S_N}$, of the trigonometric CM system in exponential coordinates and the corresponding quantum integrals $ \cL_{p}^{\mathcal H}$ of the rational CM system with harmonic term, so that the  diagram
\[
\begin{tikzcd}
\cQ_\mathcal A \arrow[r,"\chi_H"] \arrow[d,swap,"{\mathcal L}_p"] &
  \cQ_\mathcal A \arrow[d,"{\mathcal L}^H_p"] \\
\cQ_\mathcal A  \arrow[r,"\chi_H"] & \cQ_\mathcal A
\end{tikzcd}
\]
is commutative. The original Lassalle--Nekrasov correspondence is the restriction of this construction to the subspace of symmetric polynomials $\R[x_1,\ldots,x_N]^{S_N}\subset \cQ_\mathcal A.$

In Section \ref{automLN} we provide a more conceptional explanation of the Lassalle--Nekrasov correspondence by showing that it can be obtained via an automorphism of the corresponding rational Cherednik algebra.

By considering more general automorphisms in Section \ref{sechigher}, we establish a higher order analogue of the Lassalle--Nekrasov correspondence in which the trigonometric side of the correspondence is unchanged and the momentum operator corresponds to the rational operator $K = E-L_p$, where $E=\sum x_i \partial_i$ is the Euler operator and $p=\sum \xi_i^{l+1}$, $l\in \N$.
 In particular, the operator $K$ has a complete set of commuting quantum integrals. 

Note that in the case $l=1$, which corresponds to the usual Lassalle--Nekrasov correspondence, the operator $K$ is conjugated to the rational CM operator with harmonic term. In the simplest one-dimensional case with multiplicity 0 and arbitrary $l$ its polynomial eigenfunctions are known as Gould--Hopper polynomials \cite{GH62}, which were studied in various relations, in particular, in  \cite{Cha11, DLMT96, VL13}. 
We show that our generalisations of these polynomials satisfy certain bispectrality properties (cf. Duistermaat and Gr\"unbaum \cite{DG86}).

In Section \ref{secDefCM} we extend the Lassalle--Nekrasov correspondence to the deformed Calogero--Moser systems describing the interaction of two types of particles \cite{CFV98}, \cite{SV04}.

\section{Baker-Aklhiezer function related to configurations of hyperplanes}

Let $\mathcal{A}$ be a finite collection of non-collinear vectors $\alpha=(\alpha_1,\ldots,\alpha_N)\in\mathbb{R}^N$ with multiplicities $m_\alpha\in\mathbb{N}$. Consider the corresponding configuration of hyperplanes with multiplicities (which we also denote as $\mathcal A$) defined by
$$
A_m(x): = \prod_{\alpha\in\mathcal{A}}(x,\alpha)^{m_\alpha}=0.
$$
The (rational) {\it Baker--Akhiezer (BA) function} $\phi(x,\lambda)$, $x,\lambda\in\mathbb{C}^N$ associated to $\mathcal{A}$ satisfies the following two conditions \cite{CV90}:
\begin{enumerate}
\item[(I)] $\phi(x,\lambda)$ is of the form
$	\phi(x,\lambda) = P(x,\lambda)e^{(x,\lambda)}$
with $P(x,\lambda)$ a polynomial in $x$ with highest degree term equal to $A_m(x)A_m(\lambda)$;
\item[(II)] for all $\alpha\in\mathcal{A}$, we have the vanishing conditions
$$	\partial_\alpha\phi(x,\lambda) = \partial_\alpha^3\phi(x,\lambda) =\cdots= \partial_\alpha^{2m_\alpha-1}\phi(x,\lambda)\equiv 0,\quad (x,\alpha)=0,$$
where $\partial_\alpha = (\alpha,\partial/\partial x)$ is the normal derivative in direction $\alpha$.
\end{enumerate}

The existence of $\phi$ puts a strong restriction on $\mathcal A$. The known cases besides the Coxeter arrangements with invariant multiplicities include their deformed  versions ${\mathcal A}_{N-1}(p)$ and ${\mathcal C}_{N}(r,s)$ as well as some special 2-dimensional configurations, see \cite{CFV99, FJ14}.

It is known that if the BA function $\phi$ exists, then it is unique and symmetric with respect to $x$ and $\lambda$: $\phi(x,\lambda) = \phi(\lambda,x),$
see \cite{CFV99}.



The key fact \cite{CV90, CFV99, VSC93} is that if the configuration $\mathcal A$ admits the BA function $\phi(x,\lambda)$,
 then there is a homomorphism $\chi^{\mathcal A}$ mapping a quasi-invariant $q\in \mathcal Q_{\mathcal A}$ to a differential
operator 
$L_q(x,\frac{\partial}{\partial x})$ such that 
\beq{BAeigenf}
L_q \phi
(x,\lambda)=q(\lambda)\phi(x,\lambda). 
\eeq 


In particular, $\phi$ satisfies the Schr\"odinger equation
\begin{equation}\label{CMeq}
	L\phi = \lambda^2\phi, \quad \lambda^2 = \sum_{i=1}^N \lambda_i^2,
\end{equation}
for the
CM operator
\begin{equation}\label{CM}
	L = \Delta - \sum_{\alpha\in\mathcal{A}}\frac{2m_\alpha}{(\alpha,x)}\partial_\alpha, \quad \Delta = \sum_{i=1}^N \partial_i^2.
\end{equation}
If $\mathcal A$ is a Coxeter system, then $L$ is gauge-equivalent to the generalised CM operator from Olshanetsky--Perelomov \cite{OP83} in the special case of integer parameters. 

In the general case the BA function $\phi(x,\lambda)$ can be given in terms of the operator $L$ by Berest's formula \cite{Ber98}
\begin{equation}\label{Berest}\nonumber
	\phi(x,\lambda) = \big(2^{|m|} |m|!\big)^{-1}\big(L-\lambda^2\big)^{|m|}\left(A_m(x)^2e^{(\lambda,x)}\right),
\end{equation}
where $|m|=\sum_{\alpha\in\mathcal{A}}m_\alpha$, see Theorem 3.1 in \cite{CFV99}.

In the Coxeter case Etingof and Ginzburg \cite{EG02A} observed that the function $\phi(x,\lambda)$  has the expansion
\begin{equation}
\label{10}
\phi(x,\lambda) =\sum_{i=0}^\infty   {q_i}(x)q^i(\lambda),
\end{equation}
where $q^i, \, i\in \mathbb Z_{\geq 0}$, is a homogeneous basis of quasi-invariants, and $q_i$ is the dual basis with respect to a natural bilinear form $\langle \cdot, \cdot \rangle_m$ (see formulae \eqref{canform}--\eqref{bform2} in Section~\ref{sec3}). The same statement and proof hold true in the general case provided that $\phi(0,0)\ne 0$, which is satisfied in all the known cases (see \cite{FV03}, \cite{FHV13}).

Multiplication of the left-hand side of   formula \eqref{10} by the Gaussian factor $\exp(-\lambda^2/2)$ remarkably leads to a very interesting quasi-invariant version of Lassalle's multivariable Hermite polynomials, see formula (\ref{AH}) above. 
 
In the next section we study these polynomials in the simplest one-dimensional case, leaving the general case for Section~\ref{sec3}.

\section {$m$-Hermite polynomials on the line}
\label{m-line}

The classical Hermite polynomials $H_n(x)$ (in the ``probabilistic" convention) are the monic orthogonal polynomials with Gaussian weight $w(x)=e^{-x^2/2}$ (see e.g. \cite{Sze39}). They satisfy the recurrence relation
$$H_{n+1}(x)=xH_n(x)-nH_{n-1}(x)$$ 
and the differential equation
$$
(-\frac{d^2}{d x^2}  + x \frac{d}{d x}) H_n = n H_n.
$$
The corresponding functions $\psi_n(x) = e^{-x^2/4} H_n(x)$
are the eigenfunctions of the quantum harmonic oscillator
$$
(-\frac{d^2}{d x^2} +\frac{x^2}{4}) \psi_n = (n+\frac12) \psi_n.
$$
The Hermite polynomials can be given by the formulas 
\begin{equation}
\label{exp1}H_n(x)=(-1)^n e^{x^2/2}\frac{d^n}{dx^n}e^{-x^2/2},
\end{equation}
\begin{equation}
\label{exp2}
H_n(x)=e^{-\frac12 \frac{d^2}{d x^2}} x^n.   
\end{equation}
Here are the first few of them:
$$H_0(x)	=	1,\,
H_1(x)	=	x,\,
H_2(x)	=	x^2-1,\,
H_3(x)	=	x^3-3x,$$
$$H_4(x)	=	x^4-6x^2+3,\,
H_5(x)	=	x^5-10x^3+15x.
$$

It will be important for us that Hermite polynomials can also be defined by the following generating series \cite{Sze39}:
\begin{equation}
\label{firsteq}
e^{k x- k^2/2} = \sum_{n=0}^\infty H_n(x) \frac{k^n}{n!}.
\end{equation}

Let us define now their $m$-version, where $m\in \Z_{\geq 0}$. In our geometric setting they correspond to the configuration $\mathcal A$ consisting of the point $0$ taken with multiplicity $m.$ 
Consider the corresponding Schr\"odinger operator
\begin{equation}
\label{CM1}\nonumber
{\mathcal H}_m = -\frac{d^2}{d x^2} +\frac{x^2}{4}  +\frac{m(m+1)}{x^2},
\end{equation}
which is the simplest case of a CM operator with harmonic term.
It is well known that it is the result of Darboux transformations applied to the harmonic oscillator ${\mathcal H}_0$ at the first $m$ odd levels.

Namely, let $W_{1, 3,\ldots, 2m-1}$ be the Wronskian of the first $m$ odd Hermite polynomials
$$
W_{1, 3,\ldots, 2m-1} = W[H_1(x), H_3(x), \ldots, H_{2m-1}(x)].
$$
Then it is easy to check that
$$
W_{1, 3,\ldots, 2m-1} = W[x, x^3, \ldots, x^{2m-1}]= \prod_{j=1}^{m-1} (m-j)!  \, 2^\frac{m(m-1)}{2} x^\frac{m(m+1)}{2}.
$$
Introduce the notation 
$$
W_{1, 3,\ldots, 2m-1}[f]: = W[x, x^3, \ldots, x^{2m-1}, f] = W[H_1, H_3, \ldots, H_{2m-1}, f ]
$$
and the differential operator $D_m$ by
\begin{equation}
\label{intert1}\nonumber
D_m(f) = \frac{W_{1,3,\ldots, 2m-1}[f e^{x^2/4}]}{W_{1,3,\ldots, 2m-1}} e^{-x^2/4}.
\end{equation}
Then one can check that the following intertwining relation holds:
\begin{equation}
\label{intert2}
\left( {\mathcal H}_m+m\right) D_m = D_m {\mathcal H}_0.
\end{equation}

Define now the {\it $m$-Hermite polynomials} on the line by
\begin{equation}
\label{mdef}
H_n^{(m)}(x) =\frac{1}{ \prod_{j=1}^{m-1} (m-j)!} (2x)^{-\frac{m(m-1)}{2}} W_{1, 3, \ldots, 2m-1}[H_n].
\end{equation}
From the intertwining relation (\ref{intert2}) it follows that 
\beq{rr1}
\psi^{(m)}_n:= D_m(\psi_n) = x^{-m}H_n^{(m)}e^{-x^2/4}
\eeq
 is an eigenfunction of the CM operator ${\mathcal H}_m$:
\beq{eigenf0}
{\mathcal H}_m \psi_n^{(m)} = (n - m+\frac12) \psi_n^{(m)}. 
\eeq

The operator (\ref{CM}) in this case has the form
$$
L_m=\frac{d^2}{dx^2}-\frac{2m}{x} \frac{d}{dx}.
$$
The rational BA function, satisfying the equation
$$
L_m\phi_m(x,\lambda)=\lambda^2 \phi_m(x,\lambda),
$$
can be written as 
\begin{equation}
\label{ba1}\nonumber
\phi_m(x,\lambda) = \frac{x^m W_{1, 3, \ldots, 2m-1}[e^{ \lambda x}]}{W_{1, 3, \ldots, 2m-1}},
\end{equation}
(see \cite{CFV99}), or, more explicitly, 
\begin{equation}
\label{ba2}
\phi_m(x,\lambda)= (x \frac{d}{d x} - 2m+1)(x \frac{d}{d x} - 2m +3)\ldots (x \frac{d}{d x} - 1)e^{\lambda x}.
\end{equation}

\begin{theorem} $m$-Hermite polynomials have the generating function
\label{thm1d1} 
\beq{genfn}\nonumber
\phi_m(x,k) e^{-k^2/2} = \sum_{n=0}^\infty H_n^{(m)}(x) \frac{k^n}{n!}.
\eeq
\end{theorem}

\begin{proof}
Consider the operator $\Phi_m$ given by $\Phi_m [f] = \frac{x^m W_{1, 3, \ldots, 2m-1}[f]}{W_{1, 3, \ldots, 2m-1}}$. The statement follows by applying the operator $\Phi_m$ to the equality \eqref{firsteq}.
\end{proof}

More properties of these polynomials are given by the following

\begin{proposition}\label{propm}
$m$-Hermite polynomials $H_n^{(m)}$ have the following properties:
\begin{enumerate}
\item
\label{prop0}
$
H_n^{(m)}(-x) = (-1)^n H_n^{(m)}(x), 
$
\item
\label{prop1}
$
H_1^{(m)} = H_3^{(m)} = \ldots =  H_{2m-1}^{(m)} = 0, 
$
\item
\label{reqeq}
${\mathcal L}_m H_n^{(m)} = n H_n^{(m)}$,
where 
$ {\mathcal L}_m = -\frac{d^2}{d x^2} + x \frac{d}{ dx} +\frac{2m}{x} \frac{d}{d x},$
\item
\label{expprop}
$
H_n^{(m)}(x)=c_{m,n}e^{-L_m/2}x^n,\, 
$
where $c_{m,n} = \prod_{k=1}^m (n-2k+1)$,
\item
\label{integrality}
$
H_n^{(m)} =c_{m,n} p_n^{(m)}(x), 
$
where $p_n^{(m)}(x)\in \Z[x]$ is a monic polynomial of degree~$n$, 
\item
$
\langle H_n^{(m)}(x)\colon n\in \Z_{\ge 0} \rangle = Q_m,
$
where $Q_m = \langle 1, x^2, x^4, \ldots, x^{2m}, x^{2m+1}, \ldots \rangle$ is the corresponding algebra of quasi-invariants.
\end{enumerate}
\end{proposition}

\begin{proof}
Properties \eqref{prop0} and \eqref{prop1} follow from the definition of $m$-Hermite polynomials \eqref{mdef}.
Property \eqref{reqeq} follows from formulas  \eqref{rr1}, \eqref{eigenf0} and the 
relation  
$$
\mathcal{L}_m = x^me^{x^2/4} (\mathcal{H}_m+m-\frac{1}{2}) x^{-m}e^{-x^2/4}.
$$
It follows from more general Propositions \ref{Cor:eigValEq}, \ref{heat} below that $H_n^{(m)}$ is proportional to $e^{-L_m/2}x^n$. It is easy to see from Theorem \ref{thm1d1} that the coefficient of proportionality equals $c_{m,n}$, which proves property \eqref{expprop}.

Consider a monic polynomial $p(x)$ of degree $n$, it has  the form
$$
p(x) = \sum_{i=0}^{n} a_i x^{n-i}
$$
for some $a_i\in \C$ with $a_0=1$. We have to show that the equality
\beq{reqeqp}
{\mathcal L}_m p = n p
\eeq
implies that $a_i\in \Z$ for all $i$.
It is easy to see that $a_i=0$ if $i$ is odd. Also the equality \eqref{reqeqp}  implies the following recurrence relation for any even $i$, $0\le i \le n$:
\beq{recc}
(n-i)(n-i-1-2m)a_i = - (i+2) a_{i+2},
\eeq
where we put $a_{n+1}=a_{n+2}=0$. By applying relation \eqref{recc} iteratively we get for $i=2k$, $k\in \Z$, that
\beq{successive}
\prod_{j=1}^k (n-2(j-1)) \prod_{j=1}^k (n-2m-1-2(j-1)) = (-1)^k 2^k k! a_{2k}.
\eeq
Depending on the parity of $n$ one of the two products in the left-hand side of equality \eqref{successive} is the product of $k$ successive even integers. If $n$ is even then the first product divided by 
$2^k k!$ equals $\binom{n/2}{k}$. If $n$ is odd then  the second product divided by
$2^k k!$ equals $\binom{(n-2m-1)/2}{k}$, where $\binom{x}{k}:=x(x-1)\ldots (x-k+1)/k!$ is integer for all integer $x$. This implies that $a_i\in \Z$ for all $i$.

As to the final property the quasi-invariance of $H_n^{(m)}$ follows from Theorem~\ref{thm1d1} and property (2) of the BA function. The statement that these polynomials generate the space of quasi-invariants $Q_m$ follows from property \eqref{integrality}.
\end{proof}

We note that integrality property \eqref{integrality} also follows from more general results in \cite{BHSS18}.

\begin{remark}
Arithmetical properties of Hermite polynomials were studied by Schur in remarkable papers \cite{Sch29, Sch31}. 
In particular, he proved that all polynomials $H_{2k}$ and $x^{-1}H_{2k-1}$  (with the exception of $H_2(x)=x^2-1$) are irreducible over $\mathbb Q$. It would be interesting to investigate  similar problems for our $m$-Hermite polynomials. 

We also note that combinatorial and representation theoretic interpretations of the coefficients of general Wronskians of Hermite polynomials were obtained recently in \cite{BDS20}.
\end{remark}

Here are the first few examples of $m$-Hermite polynomials for $m=1$:
$$
H_0^{(1)} = -1, \, H_1^{(1)} = 0, \, H_2^{(1)} = x^2 +1, \, H_3^{(1)} = 2 x^3,  \, H_4^{(1)} = 3 ( x^4 - 2 x^2 -1) \,
$$
$$
H_5^{(1)} = 4 x^3 ( x^2 - 5), \, H_6^{(1)} = 5 (x^6- 9 x^4 + 9 x^2  +  3), \, H_7^{(1)} = 6 x^3 ( x^4  - 14 x^2 + 35),
$$
$$
H_8^{(1)} = 7 ( x^8 - 20 x^6 + 90 x^4 - 60 x^2  -15), \,  H_9^{(1)} = 8 x^3 ( x^6 - 27 x^4 + 189 x^2  -315),
$$
$$
H_{10}^{(1)} = 9 (x^{10}- 35 x^8  + 350 x^6  - 1050 x^4  + 525 x^2  + 105).
$$

Note that $m$-Hermite polynomials $H_{2k+1}^{(m)}$, $k\in \Z_{\ge 0}$,  are odd and divisible by $x^{2m+1}$. Indeed, they are quasi-invariant and $x^{2m+1}$ is the odd quasi-invariant of the lowest degree.  

Let us define the monic polynomials $E_k^{(m)}(x)$ and $G_k^{(m)}(x)$ of degree $k$, $k\in \Z_{\ge 0}$, such that 
\beq{EGpol}\nonumber
c_{m, 2k} E_k^{(m)}(x^2) = H_{2k}^{(m)} (x), \quad
c_{m, 2k+2m+1} x^{2m+1}G_k^{(m)}(x^2) = H_{2k+2m+1}^{(m)} (x).
\eeq
 Let us recall that the generalised Laguerre polynomials $L_n^{(\alpha)}$ are polynomials of degree $n$ with the highest order coefficient $(-1)^n/n!$, which satisfy the differential equation
\beq{Lag}
z \frac{d^2 L_n^{(\alpha)}(z) }{d z^2} +(\alpha + 1 -z) \frac{d L_n^{(\alpha)}(z)}{d z} + n L_n^{(\alpha)}(z)=0.
\eeq
We have the following relation of $m$-Hermite polynomials with the generalised  Laguerre polynomials for special values of the parameter $\alpha$ (cf. \cite{BDS20} where general Wronskians of Hermite polynomials were expressed via Wronskians of Laguerre polynomials).
\begin{proposition}
We have 
$$
E_n^{(m)}(x) = (-1)^n n! L^{(\alpha)}_n(\frac12 x^2), \quad \alpha=-m-\frac12, \,
$$
and
$$
G_n^{(m)}(x) = (-1)^n n! L^{(\alpha)}_n(\frac12 x^2), \quad \alpha=m+\frac12, 
$$
for any  $n \in \Z_{\ge 0}$.
\end{proposition}
\begin{proof}
Let us do the change of variables $z=\frac12 x^{2}$ in equation \eqref{Lag}. We get the following equation:
\beq{Lag2}
 \frac{d^2 L_n^{(\alpha)}(x^2) }{d x^2} +\frac{2\alpha +1}{x} \frac{d L_n^{(\alpha)}(x^2)}{d x}  -  x \frac{d L_n^{(\alpha)}(x^2)}{d x}+ 2n L_n^{(\alpha)}(x^2)=0.
\eeq
which coincides with the equation in property \eqref{reqeq} of Proposition \ref{propm} for $H_{2n}(x)$  when  $\alpha=-m-\frac12$. 

Let us conjugate the equation \eqref{Lag2} by $x^{2\alpha}$. It follows that $f=x^{2\alpha} L_n^{(\alpha)}(x^2)$ satisfies the equation
\beq{Lag3}\nonumber
 \frac{d^2 f }{d x^2} +\frac{1-2\alpha}{x} \frac{d f}{d x}  -  x \frac{d f}{d x}+ (2\alpha +2n) f=0,
\eeq
which coincides with the equation in property \eqref{reqeq} of Proposition \ref{propm} for $H_{2n+2m+1}(x)$  when  $\alpha=m+\frac12$. 
\end{proof}
Recall that generalised Laguerre polynomials are orthogonal with respect to the bilinear form 
$$
(f, g) = \int_{0}^\infty x^{\alpha} f(x) g(x) e^{-x} dx.
$$
Note that the integral in general diverges for $\alpha\le-1$ but it can be regularised using complex contour integrals. A related problem for the corresponding $m$-Hermite polynomials was discussed in a wider context of exceptional Hermite polynomials in   \cite{HHV16}.

\section{$\mathcal A$-Hermite polynomials: definition and properties}
\label{sec3}

%
%

Let us assume now that the configuration $\mathcal A$ is such that the rational BA function $\phi(x,\lambda)$ exists and satisfies the condition $\phi(0,0)\ne 0$. 
This function is both symmetric and homogeneous in the sense that
\begin{equation}\label{sym}
\phi(x,\lambda)=\phi(\lambda,x)
\end{equation}
and
\begin{equation}\label{hom}
\phi(tx,\lambda)=\phi(x,t\lambda),\ \ \ t\in\C,
\end{equation}
respectively. Symmetry follows from Theorem 2.3 in \cite{CFV99}, whereas homogeneity is a simple consequence of the defining conditions and uniqueness.

It also satisfies the following integral identity
\begin{equation}\label{intId}
\int_{i\xi+\R^N}\frac{\phi(-iz,\lambda)\phi(iz,\mu)}{A_m(z)^2}d\gamma(z)=e^{-(\lambda^2+\mu^2)/2}\phi(\lambda,\mu),
\end{equation}
where $d\gamma(z)=(2\pi)^{-N/2}e^{-z^2/2}dz$, 
and $\xi\in\R^N$ is such that $A_m(\xi)\neq 0$,
which was established in \cite{FHV13}.

 For each homogeneous $q\in\mathcal Q_{\mathcal A}$, there exists a unique homogeneous differential operator
$$
L_q=q(\partial_{1},\ldots,\partial_{N})+\mathrm{l.o.t.}
$$
of homogeneity degree $-\deg q$ such that $[L_q,L]=0$. Such operators form a commutative ring isomorphic to $\mathcal Q_{\mathcal A}$, with $L$ corresponding to $x^2$ and the ring isomorphism given by
\begin{equation}\label{LqEq}
L_q\phi(x,\lambda)=q(\lambda)\phi(x,\lambda).
\end{equation}
Moreover, as it follows from \cite{Ber98}, \cite{Cha98}, \cite{CFV99}, 
\beq{new35}
L_q\mathcal Q_{\mathcal A}\subset\mathcal Q_{\mathcal A},\ \ \ \forall q\in\mathcal Q_{\mathcal A}.
\eeq

Attached to $\mathcal Q_{\mathcal A}$ is a natural bilinear form, given by
\begin{equation}
\label{canform}
(p,q)_m=(L_p q)(0),\ \ \ p,q\in\mathcal Q_{\mathcal A},
\end{equation}
which can be seen to be symmetric \cite{EG02A}, \cite{FV03}.
Then Theorem 3.1 in \cite{FHV13} yields the integral representation
$$
(p,q)_m=\phi(0,0)\int_{i\xi+\R^N}\frac{(e^{L/2}p)(-ix)(e^{L/2}q)(ix)}{A_m(x)^2}d\gamma(x).
$$

To proceed further, we find it convenient to renormalise the bilinear form $(\cdot,\cdot)_m$: 
we set
\begin{equation}\label{bform2}
\langle p,q\rangle_m=\phi(0,0)^{-1}(p, q)_m,\ \ \ p,q\in \mathcal Q_{\mathcal A}.
\end{equation}
Choosing a homogeneous basis $q_i$, $i\in\Z_{\ge 0}$, in $\mathcal Q_{\mathcal A}$, the defining condition (II) for $\phi$ entails that
\begin{equation}\label{phiExp}
\phi(x,\lambda)=\sum_{i=0}^\infty q_i(x)q^i(\lambda)
\end{equation}
for some polynomials $q^i$. By symmetry \eqref{sym} and homogeneity \eqref{hom}, we have $q^i\in\mathcal Q_{\mathcal A}$ and $\deg q^i=\deg q_i$. Etingof and Ginzburg \cite{EG02A} showed that the quasi-invariants $q^i$, $i\in\Z_{\ge 0}$, constitute a dual basis in $\mathcal Q_{\mathcal A}$:
\begin{equation}\label{dual}
\langle q_i,q^j\rangle_m=\delta_{ij}.
\end{equation}
(Strictly speaking, they considered only the Coxeter case but their arguments work in our more general case as well, see also \cite{FV03}.)

Let us introduce now the analogues of the Hermite polynomials in the algebra of quasi-invariants $\mathcal Q_{\mathcal A}$.
Let 
\begin{equation}\label{FExp}
F(x,\lambda) =  \phi(x,\lambda)\exp(-\lambda^2/2).
\end{equation}
We define a linear ``Hermitisation" map
$\chi_H:\mathcal Q_{\mathcal A}\to \mathcal Q_{\mathcal A}, q\mapsto H_q$, by
\begin{equation}\label{HqDef2}
H_q(x)=\langle F(x,\cdot),q(\cdot)\rangle_m,\ \ \ q\in\mathcal Q_{\mathcal A}.
\end{equation}
Polynomials $H_q(x)$ for homogeneous $q\in \mathcal Q_{\mathcal A}$ are called {\it $\mathcal A$-Hermite polynomials}. For a quasi-invariant polynomial $q$ of degree $d$ the polynomial $H_q(x)$ has the form
\begin{equation}\label{HqExp}
H_q(x)=q(x)+\sum_{n=1}^{\lfloor d/2\rfloor}H_q^{(d-2n)}(x),
\end{equation}
with $H_q^{(d^\prime)}\in\mathcal Q_{\mathcal A}[d^\prime]$, where   $\mathcal Q_{\mathcal A}[d^\prime]$ denotes
the space of  homogeneous quasi-invariants of degree $d'$. 

For the basis $q_i$ of quasi-invariants $\mathcal Q_{\mathcal A}$
we can define the corresponding $\mathcal A$-Hermite polynomials by the generating function expansion
\begin{equation}\label{FExp2}
F(x,\lambda)=\sum_{i=0}^\infty H_{q_i}(x)q^i(\lambda).
\end{equation}

\begin{proposition}
The $\mathcal A$-Hermite polynomials $H_{q_i}(x)$ form a basis in the space of quasi-invariants $\mathcal Q_{\mathcal A}$. 
\end{proposition}
\begin{proof}
Quasi-invariance of the polynomials  $H_{q_i}(x)$ follows from formulae \eqref{FExp}, \eqref{FExp2} and property (II) of the BA function. It is also clear from formula \eqref{HqExp} that these polynomials are linearly independent and span the space $Q_{\mathcal A}$.
\end{proof}

Now we establish some properties of $\mathcal A$-Hermite polynomials.
Let us introduce the Euler operator
\begin{equation*}
E= E_x =  \sum_{i=1}^N x_i\partial_i
\end{equation*}
and the following ``harmonic" version of the operator (\ref{CM}):
\begin{equation}
\label{harm}
\mathcal L^{\mathcal H}= L-E =\Delta - \sum_{\alpha\in\mathcal{A}}\frac{2m_\alpha}{(\alpha,x)}\partial_\alpha-\sum_{i=1}^N x_i\partial_i.
\end{equation}

\begin{proposition}\label{Cor:eigValEq}
For $q\in\mathcal Q_{\mathcal A}[d]$, $H_q(x)$satisfies the differential equation
\begin{equation}\label{HqEq}
\mathcal L^{\mathcal H}H_q=-d H_q.
\end{equation}
\end{proposition}

\begin{proof}
Taking $t\to 0$ in the identity $(\phi(tx,\lambda)-\phi(x,t\lambda))/t=0$ (see \eqref{hom}), we obtain
\begin{equation*}
E_x\phi(x,\lambda)-E_\lambda\phi(x,\lambda)=0.
\end{equation*}
When combining this identity with equation \eqref{CMeq}, we readily find that $F(x,\lambda)$ satisfies the differential equation
\begin{equation}\label{genFuncEq}\nonumber
\mathcal L^{\mathcal H}_xF=-E_\lambda F.
\end{equation}
Note that the operator $E$ is self-adjoint with respect to $\langle\cdot,\cdot\rangle_m$ since homogeneous quasi-invariants of different degrees are orthogonal.  
Hence, for any $q\in\mathcal Q_{\mathcal A}$, we get
\begin{equation*}
\big(\mathcal L^{\mathcal H} \chi_H\big)(q) = \big\langle \mathcal L^{\mathcal H}_xF(x,\cdot),q(\cdot)\big\rangle_m
=-  \langle F(x,\cdot),Eq (\cdot)\rangle_m
= - (\chi_H  E)(q),
\end{equation*}
which implies the statement.
\end{proof}

\begin{proposition}\label{unique}
For $q\in\mathcal Q_{\mathcal A}[d]$, $H_q(x)$ is the unique quasi-invariant polynomial of the form \eqref{HqExp} that satisfies the differential equation \eqref{HqEq}.
\end{proposition}

\begin{proof}
Let us treat \eqref{HqExp} as an ansatz. From \cite{Cha98}, \cite{CFV99} we derive that 
$$
L(\mathcal Q_{\mathcal A}[d])\subset\mathcal Q_{\mathcal A}[d-2], \quad  d\geq 2. 
$$
Substituting the right-hand side of \eqref{HqExp} for $H_q$ in \eqref{HqEq}, and using 
$$
EH_q^{(d-2n)}=(d-2n)H_q^{(d-2n)},
$$
we arrive at the recurrence relation
\begin{equation*}
2nH_q^{(d-2n)}= - LH_q^{(d-2n+2)},\ \ \ n=1,\ldots\lfloor d/2\rfloor,
\end{equation*}
which implies uniqueness of the solution $H_q(x)$ such that $H_q^{(d)}=q$.
\end{proof}

The next proposition details an integral representation for $H_q$, obtained as a straightforward consequence of formula \eqref{FExp} and the integral identity \eqref{intId}.

\begin{proposition}\label{Prop:intRep}
Let $q\in\mathcal Q_{\mathcal A}[d]$. Then we have
\begin{equation}\label{intRep}
H_q(x)=\exp(x^2/2)\int_{i\xi+\R^N}\frac{q(-iz)\phi(iz,x)}{A_m(z)^2}d\gamma(z),
\end{equation}
where $d\gamma(z)=(2\pi)^{-N/2}e^{-z^2/2}dz$ and $\xi \in \R^N$ satisfies $A_m(\xi)\ne 0$ as before.
\end{proposition}

\begin{proof}
By linearity in $q$,  it suffices to prove \eqref{intRep} for $q=q_i$, $i\in\Z_{\ge 0}$. Multiplying \eqref{intId} by $\exp(\mu^2/2)$ and taking $\mu\to x$, we obtain
\begin{equation}
\label{rr11}
F(x,\lambda)=\exp(x^2/2)\int_{i\xi+\R^N}\frac{\phi(-iz,\lambda)\phi(iz,x)}{A_m(z)^2}d\gamma(z).
\end{equation}
Due to formula \eqref{phiExp}  we have
$$
\phi(-i z,\lambda)=\sum_{j=0}^\infty q_j(-i z)q^j(\lambda).
$$
By substituting this expansion into formula \eqref{rr11} and comparing it with the expansion \eqref{FExp2} for $F(x, \lambda)$ we get the claim.
\end{proof}

\begin{proposition}\label{Prop:diffRep}
For $q\in\mathcal Q_{\mathcal A}[d]$, we have
\begin{equation*}
H_q(x)=(-1)^d\exp(x^2/2)L_q\exp(-x^2/2).
\end{equation*}
\end{proposition}

\begin{proof}
By using formula  \eqref{LqEq}  and  
the fact that $q\in\mathcal Q_{\mathcal A}[d]$, we rewrite \eqref{intRep} as
\begin{equation*}
H_q(x)=(-1)^d\exp(x^2/2)L_q\int_{i\xi+\R^N}\frac{\phi(iz,x)}{A_m(z)^2}d\gamma(z).
\end{equation*}
From \eqref{phiExp} it is evident that $\phi(x,0)=\phi(0,\lambda)=\phi(0,0)$. Setting $\lambda=0$ in the identity \eqref{intId}, we thus find that
\begin{equation*}
\int_{i\xi+\R^N}\frac{\phi(iz,x)}{A_m(z)^2}d\gamma(z)=\exp(-x^2/2),
\end{equation*}
and the claim follows.
\end{proof}

The polynomials $H_q$ can also be written in the following form (cf. Theorem 3.4 in R\"osler \cite{Ros98}).

\begin{proposition}\label{heat}
For $q\in\mathcal Q_{\mathcal A}[d]$, we have
\begin{equation*}
H_q(x)=e^{-L/2} q.
\end{equation*}
\end{proposition}

\begin{proof}
Note that $[E, L]=- 2 L$, which implies 
$
[E, e^{-L/2}] = L e^{-L/2}.
$
Therefore 
$$
\mathcal L^{\mathcal H}e^{-L/2} q=(L-E)  e^{-L/2} q = - e^{-L/2} E q = -d (e^{-L/2} q),
$$
and the statement follows by Proposition \ref{unique}.
\end{proof}

As the following theorem demonstrates, the Hermitisation map $\chi_H$ becomes a unitary map when 
 the domain has the bilinear form $\langle\cdot,\cdot\rangle_m$ and the codomain is equipped with the bilinear form
$$
\{p,q\}_m\equiv \int_{i\xi+\R^N}\frac{p(z)q(z)}{A_m(z)^2}d\gamma(z),\ \ \ p,q\in\mathcal Q_{\mathcal A}.
$$

\begin{theorem}\label{two_forms}
For $p,q\in\mathcal Q_{\mathcal A}$, we have
\begin{equation*}
\{H_p,H_q\}_m =\langle p,q\rangle_m.
\end{equation*}
\end{theorem}

\begin{proof}
By linearity, we may and shall assume $p\in\mathcal Q_{\mathcal A}[d]$ and $q\in\mathcal Q_{\mathcal A}[d^\prime]$ for some $d,d^\prime\in\Z_{\ge 0}$. From the integral identity \eqref{intId}, and using property \eqref{hom} with $t=\pm i$, we infer
\begin{equation}
\label{inteq}
\int_{i\xi+\R^N}\frac{F(z,-i\lambda)F(z,i\mu)}{A_m(z)^2}d\gamma(z)=\phi(\lambda,\mu).
\end{equation}
Note that
\beq{pmh}
\langle p(\nu), F(z, \pm i \nu) \rangle_m = (\pm i)^d \langle p(\nu), F(z, \nu)\rangle_m.
\eeq
Therefore by
applying $\langle p,\cdot\rangle_m$ in $\lambda$ to the equality \eqref{inteq}, it follows from fornulas \eqref{HqDef2} and  \eqref{BAeigenf}  that
$$
(-i)^d\int_{i\xi+\R^N}\frac{H_p(z)F(z,i\mu)}{A_m(z)^2}d\gamma(z)=p(\mu).
$$
The claim now follows by applying $\langle q,\cdot\rangle_m$ in $\mu$ to this equality and using again formulae \eqref{HqDef2} and \eqref{pmh}.
\end{proof}

\section{The Lassalle--Nekrasov correspondence and multivariable $m$-Hermite polynomials}
\label{sect5}

In this section, we consider the case of the Coxeter configuration $\mathcal A$ of type $A_{N-1}$ in which $\mathcal A$ is the set of positive roots ${A_{N-1}}_+=\{e_i - e_j: 1\le i<j\le N\}\subset \R^N$ with all the roots taken with the same multiplicity $m\in \Z_{\ge 0}$. We will call the corresponding $\mathcal A$-Hermite polynomials {\it multivariable $m$-Hermite polynomials}. 

The one-dimensional $m$-Hermite polynomials  from Section \ref{m-line} correspond to a different embedding of the root system $A_1\subset \R$. Moreover, as explained in further detail at the end of this section, Lassalle's multivariable Hermite polynomials \cite{Las91} with parameter $\alpha = -1/m$ can be viewed as special instances of multivariable $m$-Hermite polynomials.


Throughout this section, we use the standard notation
$$
\Lambda_N = \C[x_1,\ldots,x_N]^{S_N}
$$
for the algebra of   symmetric polynomials in $N$ variables $x_1,\ldots,x_N$.

In the present case, the operator (\ref{harm}) takes the form
\begin{equation}
\label{LH}
\mathcal  L^{\mathcal H}=\Delta - \sum_{i<j}^N\frac{2m}{(x_i-x_j)}(\partial_i-\partial_j)-\sum_{i=1}^Nx_i\partial_i
\end{equation}
and is conjugated to the $N$-particle CM operator
$$
L^{\mathcal H}=\Delta-\sum_{i<j}^N\frac{2m(m+1)}{(x_i-x_j)^2}-\frac{1}{4}\sum_{i=1}^N x_i^2,
$$
$$
\mathcal{L}^{\mathcal H} =  \widehat \Psi_0^{-1} (L^{\mathcal H}  - \frac12 mN(N-1)+\frac{N}{2}) \widehat\Psi_0,
$$
where $\widehat \Psi_0$ is the operator of multiplication by $\Psi_0=\prod_{i<j}^N(x_i-x_j)^{-m}\exp(-x^2/4).$

It can be included in a commutative ring of differential operators $\mathcal L^{\mathcal H}_p, \, p\in\Lambda_N$, with highest order terms $p(\partial_1^2,\dots, \partial_N^2)$ (see \cite{OP83, Poly92}), given explicitly by formula (\ref{cLpH}) below.

Lassalle \cite{Las91} (in the quantum case) and Nekrasov \cite{Nek97} (both in the classical and  quantum cases) discovered a remarkable correspondence between this system and Sutherland's version of the CM system with the Hamiltonian
$$
H_N = -\Delta+\sum_{i<j}^N \frac{2m(m+1)}{\sin^2(z_{i}-z_{j})}.
$$ 
It has an eigenfunction $\Phi_0= \prod_{i<j}^N \sin^{-m} (z_i-z_j)$ with eigenvalue $\lambda_0= m^2N(N^2-1)/3$. Its gauged version $\frac{1}{4}\widehat\Phi_0^{-1} (H_N-\lambda_0) \widehat\Phi_0$, written in the exponential coordinates 
$x_j = e^{2i z_j}$, takes the form
\begin{equation}
\label{CMtrig}
 {\mathcal H}_{N}=\sum_{i=1}^N
 (x_{i}{\partial_{i}})^2-m\sum_{ i < j}^N
\frac{x_{i}+x_{j}}{x_{i}-x_{j}}(x_i\partial_{i}-x_j
\partial_{j} ).
\end{equation}


To construct the corresponding quantum integrals, Heckman \cite{Hec91} introduced the following version of the Dunkl operators:
 \begin{equation}
  \label{heckdun}\nonumber
 \cD_{i}=x_i {\partial_i}-\frac{m}{2}\sum_{j\ne i}^N\frac{x_i+x_j}{x_i-x_j}(1-s_{ij}), \ \ \ i=1,\dots, N,
 \end{equation}
 where $s_{ij}$ is the transposition acting on functions by permuting the coordinates $x_i$ and $x_j.$
 They are related to the original Dunkl operators \cite{Dun89}
 $$
D_i={\partial_i}+m\sum_{j\neq i}\frac{1}{x_i-x_j}(s_{ij}-1)
$$
by
$$
\cD_i=x_iD_i -\frac{m}{2}\sum_{j\neq i}(s_{ij}-1).
$$
Heckman's operators $\cD_i$ do not commute, but the differential operators
\begin{equation}
\label{heck}
 \mathbb L_k = {\rm Res} \,(\cD_{1}^k+\dots+\cD^k_{N}),
\end{equation}
where ${\rm Res}$ means the restriction to the space of symmetric polynomials $\Lambda_N$, do commute with each other. Since $\mathbb L_2= {\mathcal H}_{N}$ they are integrals of the quantum trigonometric CM system (\ref{CMtrig}).

For our purposes, it is convenient to replace Heckman's operators ${\mathcal D}_i$ by Polychronakos'{} operators \cite{Poly92}
$$
\pi_i = x_i D_i,\ \ \ i=1,\ldots,N.
$$
The corresponding differential operators
\begin{equation}
\label{PolyInt}
{\mathcal L}_k = {\rm Res} \,(\pi_{1}^k+\dots+\pi^k_{N})
\end{equation}
pair-wise commute and $\mathcal L_2={\mathcal H}_{N}$, (see \cite[Proposition 5.2]{SV15} for the precise relation between the two sets of quantum integrals \eqref{heck} and \eqref{PolyInt}). Note that
$$
\mathcal L_1 = \mathbb L_1=x_1\partial_1+\dots+x_N\partial_N
$$
is simply the total momentum of the system (in the exponential coordinates).

More generally, for every $p\in\Lambda_N$, there exists a unique $S_N$-invariant differential operator $\cL_p$ such that
\begin{equation}\label{cLpCond}\nonumber
\cL_pq=p(\pi_1,\ldots,\pi_N)q,\ \ \ \forall q\in\Lambda_N,
\end{equation}
which commutes with all other such operators $\cL_{r}$, $r\in\Lambda_N$.
The operator (\ref{PolyInt}) corresponds  to the Newton power sum $p_k(\xi)=\xi_1^k+\dots+\xi_N^k.$

Consider now the  algebra of quasi-invariants $\mathcal Q_{\mathcal A},$ which we in this case will denote by $\cQ_m$. It consists of the polynomials $p(\xi_1,\dots,\xi_N)$ which are invariant under permutations $s_{ij}$ of $\xi_i$ and $\xi_j$ up to order $2m$:
$$
p(\xi)-p(s_{ij}\xi)\equiv 0 \mod (\xi_i-\xi_j)^{2m+1}.
$$
The symmetric algebra $\Lambda_N$ is a subalgebra of the algebra ${\mathcal Q}_m$.

First we show that the algebra $\cQ_m$ is preserved by the operators $\cL_p.$

\begin{proposition}\label{Prop:QmAction}
For each $p\in\Lambda_N$, we have
$$
\cL_p\cQ_m\subset\cQ_m.
$$
\end{proposition}
\begin{proof}
Let $H_m$ be the rational Cherednik algebra corresponding to the symmetric group $S_N$ and parameters $t=1, c=m$ \cite{EG02B}. It is generated by operators $D_i, x_i$, and $s_{ij}$. 

Consider the spherical subalgebra $eH_m e \subset H_m$, where
 $e= \frac{1}{N!}\sum_{w\in S_N}w$.
One can identify operators  $p(\pi_1,\ldots,\pi_N)$,  $p\in \Lambda_N$,  with elements 
$$
m_e(p):=p(\pi_1,\ldots,\pi_N)e \in eH_m e.
$$
Since, up to analogous identifications, $\C[x_1,\ldots,x_N]^{S_N}$ and $\C[D_1,\ldots,D_N]^{S_N}$ generate $eH_me$ \cite{BEG02}, we can write $p(\pi_1,\ldots,\pi_N)$ as a  polynomial in (non-commuting) operators of multiplication by $q(x_1,\ldots,x_N)$ and $q(D_1,\ldots,D_N)$ with $q\in\Lambda_N$.
Therefore the differential operator $\cL_p$ can be expressed  as a polynomial in $q(x)$ and $L_q$. 
The statement follows from formula \eqref{new35}.
\end{proof}

To proceed further, we need the following lemma (cf. \cite{DJO94} where a related property was studied).

\begin{lemma}\label{Lemma:nonDeg} 
The restriction of $(\cdot,\cdot)_m$ to $\Lambda_N$ is non-degenerate.
\end{lemma}

\begin{proof}
Let
$$
\Ker(m)=\{p\in \Lambda_N : (p,q)_m=0,\, \forall q\in \Lambda_N\}.
$$
For $p\in\Rad(m)$ and $q,q^\prime\in \Lambda_N$, we have
\begin{equation*}
\begin{split}
(L_qp,q^\prime)_m&=(p,qq^\prime)_m=0.
\end{split}
\end{equation*}
Hence $\Ker(m)$ is a graded ideal in $\Lambda_N$ which is invariant under the operators $L_q$. Assuming $\Rad(m)\neq\{0\}$, we choose a non-zero homogenous $p\in \Ker(m)$ of minimal degree. By the invariance of $\Ker(m)$ and minimality of $p$,
$$
L_q p=0,\ \ \ \forall q\in \Lambda_N,
$$
that is  $p$ is $m$-harmonic in the terminology of \cite{FV02}. Since   the subspace of $S_N$-invariant $m$-harmonic polynomials coincides with $\C$ we must have $p\in\C$. Then $(p,p)_m=p^2$, which implies $p=0$, and the statement follows. 
\end{proof}

This enables us to establish self-adjointness of the operators $\cL_p$ with respect to the bilinear form $(\cdot,\cdot)_m$.

\begin{proposition}\label{Prop:selfAdj}
For each $p\in \Lambda_N$, we have
\begin{equation}\label{selfAdj}
(\cL_p q,q^\prime)_m=(q,\cL_pq^\prime)_m,\ \ \ q,q^\prime\in\cQ_m.
\end{equation}
\end{proposition}

\begin{proof}
We recall that any differential operator in variables $x_1,\ldots,x_N$ that has rational coefficients with poles located only along the hyperplanes $x_i=x_j$, $1\leq i<j\leq N$, is completely determined by its action on $\Lambda_N$, see e.g.~Lemma 2.8 in \cite{Opd95}. Both $\cL_p$ and its dual $\cL_p^*$ are such differential operators. This is immediate for ${\mathcal L}_p$, and it follows for $\cL_p^*$   from the proof of Proposition \ref{Prop:QmAction} and the fact that $L_q^*=q$, $q^* = L_q$ for all $q\in \cQ_m$. By  Lemma \ref{Lemma:nonDeg}, it thus suffices to prove \eqref{selfAdj} for $q,q^\prime\in\Lambda_N$.

To this end, we recall the bilinear form \cite{DJO94}
$$
[p,q]_m := \big(p(D_1,\ldots,D_N)q\big)(0),\ \ \ p,q\in\C[x_1,\ldots,x_N],
$$
and its  properties
\beq{djoformprop}
[x_j p,q]_m=[p,D_jq]_m, \quad [D_j p,q]_m=[p,x_j q]_m,
\eeq 
for all $j=1,\ldots,N$.
It follows from relations \eqref{djoformprop}  that
\beq{self1}
[\pi_jp,q]_m=[p,\pi_jq]_m.
\eeq
Now let $q,q^\prime\in\Lambda_N$. We deduce by the property \eqref{self1} that
$$
(\cL_p q,q^\prime)_m=\big[p(\pi_1,\ldots,\pi_N)q,q^\prime\big]_m=\big[q,p(\pi_1,\ldots,\pi_N)q^\prime\big]_m=(q,\cL_pq^\prime)_m,
$$
and the claim follows.
\end{proof}

As a simple consequence, we can now prove that the BA function $\phi(x, \lambda)$ satisfies the following remarkable symmetry property.

\begin{theorem}\label{phisym}
For any $p\in\Lambda_N$, we have
\begin{equation}
\label{phiId}\nonumber
\cL_{p,x}\phi(x,\lambda)=\cL_{p,\lambda}\phi(x,\lambda),
\end{equation}
where notations ${\mathcal L}_{p, x}$ and ${\mathcal L}_{p, \lambda}$ stand for the operator ${\mathcal L}_{p}$ acting in the variables $x=(x_1,\ldots, x_N)$ and $\lambda = (\lambda_1, \ldots, \lambda_N)$, respectively.
\end{theorem}

\begin{proof}
Let $q_i$ and $q^i$ be dual bases of quasi-invariants ${\mathcal Q}_m$ so that formulas \eqref{phiExp}--\eqref{dual} hold.
Defining
$$
a_{i}^j := \langle\cL_p q_i, q^j\rangle_m, \quad  b_{i}^j:= \langle q_i,\cL_p q^j\rangle_m,
$$
we have by Proposition \ref{Prop:selfAdj} and formula \eqref{bform2} that 
$
a_{i}^j  = b_{i}^j
$
for all $i,j\in\Z_{\ge 0}$. It follows that
$$
\cL_{p,x}\phi(x,\lambda) = \sum_i q^i(\lambda)  \sum_j a_{i}^jq_j(x)  
= \sum_jq_j(x)\sum_i b_i^{j}q^i(\lambda)
= \cL_{p,\lambda}\phi(x,\lambda).
$$
\end{proof}


For $p\in\Lambda_N$, let $\cL_p=\cL_{p,x}$, and let $L$ be given by  (\ref{CM}), which in this case is
$$
L = \Delta - \sum_{i<j}\frac{2m}{x_i-x_j}(\partial_i-\partial_j).
$$
Following Baker and Forrester \cite{BF97}, we consider the differential operators 
\begin{equation}
\label{cLpH}
\cL_p^{\mathcal H}=\cL_{p}+\frac{1}{2}[\cL_{p},L]+\frac{1}{2^22!}[[\cL_{p},L],L]+\cdots+\frac{1}{2^dd!}[\cdots[\cL_{p},L],\ldots,L],
\end{equation}
where the last term contains $d=\deg p$ commutators. 

The operators $\cL_p^{\mathcal H}$ have order $ 2 \deg p$ and the constant coefficient highest order terms $(-1)^d p(\partial_1^2,\dots, \partial_N^2).$
In particular, when $p(\xi)=\xi_1+\dots+\xi_N$ the corresponding operator
$$
\cL_p^{\mathcal H}=\sum_{i=1}^N x_i\partial_i-L=-\mathcal  L^{\mathcal H}
$$
coincides up to a sign with the CM operator (\ref{LH}). As we will see below, these operators commute and thus are quantum integrals of the CM system with harmonic term. This statement can also be extracted from the work of Baker and Forrester \cite{BF97}.

Let $\chi_H\colon \cQ_m \to \cQ_m$ be the Hermitisation map defined by (\ref{HqDef2}), and $\cL_{p},\cL_p^{\mathcal H}$ the operators defined above. Note that both of them preserve the algebra of quasi-invariants due to Proposition \ref{Prop:QmAction} and the fact that $L$ preserves  $\cQ_m.$

We are now ready to state and prove our generalisation of the Lassalle-Nekrasov correspondence to the algebra of quasi-invariants $\cQ_m$.

\begin{theorem}\label{intertwine}
The Hermitisation map $\chi_H$ intertwines the trigonometric and rational  harmonic CM operators and their quantum integrals. 
More precisely, the diagram
\[
\begin{tikzcd}
\cQ_m \arrow[r,"\chi_H"] \arrow[d,swap,"{\mathcal L}_p"] &
  \cQ_m \arrow[d,"{\mathcal L}^{\mathcal H}_p"] \\
\cQ_m  \arrow[r,"\chi_H"] & \cQ_m
\end{tikzcd}
\]
is commutative for all $p\in\Lambda_N$, that is 
$$
\cL_p^{\mathcal H}\circ \chi_H = \chi_H\circ\cL_p.
$$
\end{theorem}
\begin{proof}
As before we will use subscript $x$ or $\lambda$ to indicate the variables the operators act in. 
Let us start with the following conjugation relation
\begin{multline}
\label{conj1}
e^{\lambda^2/2}\cL_{p,\lambda}e^{-\lambda^2/2}
=\cL_{p,\lambda}+\frac{1}{2}[\lambda^2,\cL_{p,\lambda}]+\frac{1}{2^22!}[\lambda^2,[\lambda^2,\cL_{p,\lambda}]] \\
  + \cdots 
+ \frac{1}{2^dd!}[\lambda^2,\ldots,[\lambda^2,\cL_{p,\lambda}]\cdots],
\end{multline}
where the last commutator is applied $d=\deg p$ times. Here we used the 
standard formula $Ad_{e^X} = e^{ad_X}$ with $X$ being the operator of multiplication by $\lambda^2/2$, and the
fact that the 
operator $ad_X$   is lowering the order of a differential operator at least by one, so that after $d=\deg p$ times the commutator becomes zero.
 
Now using formula \eqref{conj1}, Theorem \ref{phisym},  and the eigenfunction property   $L_x\phi = \lambda^2\phi$ we get
\begin{equation*}
\begin{split}
\cL_{p,\lambda}\ F(x, \lambda)&=
e^{-\lambda^2/2}(e^{\lambda^2/2}\cL_{p,\lambda}e^{-\lambda^2/2})\phi \\
&
=\Big(\cL_{p,x}+\frac{1}{2}[\cL_{p,x},L_x]+\frac{1}{2^22!}[[\cL_{p,x},L_x],L_x] \\
&
+ \cdots+\frac{1}{2^dd!}[\cdots[\cL_{p,x},L_x],\ldots,L_x]\Big)\phi\ e^{-\lambda^2/2} ={\mathcal L}^{\mathcal H}_{p, x} F(x, \lambda). 
\end{split}
\end{equation*}
Therefore by Proposition \ref{Prop:selfAdj} we get 
\begin{equation*}
\big(\cL_p^{\mathcal H} \chi_H\big)(q) = \big\langle \cL_{p,x}^{\mathcal H} F(x,\cdot),q(\cdot)\big\rangle_m
=\langle F(x,\cdot),\cL_pq(\cdot)\rangle_m
= (\chi_H\cL_p)(q),
\end{equation*}
as required.
\end{proof}

As a corollary, the commutativity of the operators $\cL_p^{\mathcal H}$ immediately follows. 

\begin{corollary}\label{cor}
The operators (\ref{cLpH}) commute:
\beq{BFcom}
[\cL_p^{\mathcal H}, \cL_q^{\mathcal H}]=0
\eeq
for all $p,q \in \Lambda_N$ and thus are quantum integrals of the CM system in an harmonic field (\ref{LH}).
\end{corollary}

Note that the coefficients of the operators $\cL_p^{\mathcal H}$ depend polynomially on $m$, and therefore commutativity \eqref{BFcom} holds for any value of the parameter $m$, not necessarily integer. 


In the quantum case, the {\it Lassalle--Nekrasov correspondence} \cite{Las91,Nek97} is given by the linear map $\chi_{LN}\colon \Lambda_N \to \Lambda_N$, 
$\chi_{LN}(p) = e^{-L/2} p$, $p\in \Lambda_N$, which is
 such that the diagram
\begin{equation}\label{LN}
\begin{tikzcd}
\Lambda_N \arrow[r,"\chi_{LN}"] \arrow[d,swap,"{\mathcal L}_p"] &
\Lambda_N \arrow[d,"{\mathcal L}^{\mathcal H}_p"] \\
\Lambda_N \arrow[r,"\chi_{LN}"] & \Lambda_N
\end{tikzcd}
\end{equation}
is commutative for all $p\in\Lambda_N$. 
Strictly speaking, the case corresponding to $p(\xi)=\xi_1+\cdots+\xi_N$ appeared in \cite{Las91}, whereas the general case can be found in \cite{BF97}. Although this formulation of the correspondence is not present in these works, 
it is readily extracted from the latter paper. In our case with $m\in \Z_{\ge 0}$ we have by Proposition \ref{heat}  the relation $\chi_H|_{\Lambda_N} = \chi_{LN}$. Moreover, our definition in \eqref{HqDef2} of the Hermitisation map $\chi_H$ is readily modified to yield $\chi_{LN}$ for non-integer $m$: the BA function should be replaced by a particular generalised hypergeometric series and $\langle\cdot,\cdot\rangle_m$ by a suitable $L^2$-inner product. 

This establishes the equivalence between the rational CM system in an harmonic field and Sutherland's version of this system, which sounds surprising, especially at the classical level. Indeed, these two systems have very different dynamics: the orbits of CM system with harmonic field are known to be periodic with the same period $2\pi$, while for the Sutherland case they are much more complicated. In fact, there is no contradiction here since the Hamiltonian ${\mathcal L}^{\mathcal H}_1$ corresponds to the momentum $\mathcal L_1$, but not to the Hamiltonian $\mathcal L_2.$


Let us demonstrate this in the simplest case $m=0$. In this case the quasi-invariant algebra $\cQ_0=\C[x_1,\ldots,x_N]$ coincides with the set of all polynomials. For any polynomial $p$ the corresponding operators are
$$
\cL_p=p(x_1\partial_1,\dots, x_N\partial_N), \quad \cL_p^{\mathcal H}=p(x_1\partial_1-\partial_1^2,\dots, x_N\partial_N-\partial_N^2).
$$
The Hermitisation $\chi_H$ sends the monomial $x^\mu=x_1^{\mu_1}\dots x_N^{\mu_N}$ to the product
$$
H_{\mu}(x)=H_{\mu_1}(x_1)\dots H_{\mu_N}(x_N)
$$
of the usual Hermite polynomials, and extends to the general polynomial by linearity. 

Note that the monomials $x^{\mu}$ are the joint eigenvectors of all operators $\cL_p$:
$$
\cL_p x^\mu=p(\mu)x^\mu,
$$
while $H_\mu(x)$ are the joint eigenvectors of $\cL_p^{\mathcal H}$:
$$
\cL_p^{\mathcal H} H_\mu(x)=p(\mu)H_\mu(x)
$$
with the same eigenvalue $p(\mu).$

On the other hand, allowing any $m\in\mathbb{N}$ while restricting attention to $\Lambda_N\subset\cQ_m$, we can recover Lassalle's multivariate Hermite polynomials. More specifically, for partitions $\lambda=(\lambda_1,\ldots,\lambda_N)$ consider the (monic) Jack polynomials $P_\lambda^{(\alpha)}(x)$ in $N$ variables $x=(x_1,\ldots,x_N)$, see e.g.~\cite{Mac95}. For us, the relevant parameter values are $\alpha = -1/m$. Note that not all Jack polynomials can be specialized to such a parameter value, e.g.
\begin{equation*}
\begin{split}
P_{(3,1)}^{(-1/m)}(x) &= m_{(3,1)}(x)+\frac{2m}{m-1}m_{(2,2)}(x)+\frac{m(5m-3)}{(m-1)^2}m_{(2,1,1)}(x)\\
&\quad +\frac{12m^2}{(m-1)^2}m_{(1,1,1,1)}(x)
\end{split}
\end{equation*}
is well-defined if and only if $m\neq 1$. (Here, the monomial symmetric polynomials $m_\lambda(x)=\sum_{a\in S_N(\lambda)}x_1^{a_1}\cdots x_N^{a_N}$.) For $m$ and $\lambda$ such that no $\alpha$-poles are encountered, we consider the symmetric $m$-Hermite polynomials
$$
H_\lambda^{(m)} := \chi_H\left(P_\lambda^{(-1/m)}\right).
$$

From Proposition \ref{unique}, we infer that $H_\lambda^{(m)}$ is the unique symmetric polynomial of the form
$$
H_\lambda^{(m)}(x) = P_\lambda^{(-1/m)}(x)+\text{lower degree terms},
$$
with lower degrees of the form $|\lambda|-2k$, $k=1,\ldots,\lfloor \frac12 |\lambda|\rfloor$, that satisfies the differential equation
$$
\cL^H H_\lambda^{(m)} = -|\lambda|H_\lambda^{(m)}.
$$
Since Lassalle's multivariate Hermite polynomials also have these two properties (see Th\'eor\`eme 1 and Corollaire to Th\'eor\`eme 3 in \cite{Las91}), they must coincide with the symmetric multivariable $m$-Hermite polynomials (up to a choice of normalisation and after setting the parameter $\alpha=-1/m$).

\section{Lassalle--Nekrasov correspondence and automorphisms of Cherednik algebras}
\label{automLN}

Recall the rational Cherednik algebra $H_m$ associated with the symmetric group $S_N$ and parameter $m$ \cite{EG02B} (we assume the parameter $t=1$). This algebra is generated by two polynomial rings $\C[x] =\C[x_1, \ldots, x_N]$, $\C[y]=\C[y_1, \ldots, y_N]$ and the group algebra $\C S_N$ with defining relations
$$
w p(x) = p(w^{-1} x) w, \quad w q(y) = q(w^{-1} y) w,\quad y_i x_j - x_j y_i = S_{ij},
$$
where $S_{ij} = m s_{ij}$ if $i\ne j$ and $S_{ii} = 1 - m \sum_{k\ne i} s_{ik}$ if $i=j$, and $p(x)\in \C[x]$, $q(y)\in \C[y]$, $w\in \C S_N$.

The algebra $H_m$ has a faithful polynomial representation $\rho\colon H_m \to {\rm End}(\C[x])$ where the elements $x_i \in H_m$ act on $\C[x]$  by multiplication,  the elements $y_i \in H_m$ act by  Dunkl operators $D_i$, and the elements $s_{ij}$ act as transpositions of $x_i$ and $x_j$.

Consider the following automorphisms  $a_{\tau,l}, b_{\tau, l}$ of the rational Cherednik algebra $H_m$  \cite{EG02B}:
\beq{automA}
a_{\tau, l}(x_i) = x_i + \tau y_i^l, \quad a_{\tau, l} (y_i) = y_i, \quad a_{\tau, l}(w)=w,
\eeq
\beq{automB}\nonumber
b_{\tau, l}(x_i) = x_i, \quad b_{\tau, l}(y_i) =  y_i + \tau x_i^l, \quad  b_{\tau, l}(w)=w,
\eeq
where $\tau \in \C$, $l \in \Z_{\ge 0}$, $w \in S_N$, and $1\le i \le N$.

Let $\mathbb G$ be the group of automorphisms of the rational Cherednik algebra $H_m$ generated by these elements $a_{\tau,l}, b_{\tau, l}$. Note that elements of $\mathbb G$ preserve the spherical subalgebra $eH_me \subset H_m$, where $e=\frac{1}{N!}\sum_{w\in S_N} w$. 

Consider the subalgebra of $H_m$ generated by the elements $\sigma_k = e \sum_{i=1}^N \pi_i^k$, where $\pi_i = x_i y_i$, $k\in \Z_{\ge 0}$ become Polychronakos operators \cite{Poly92} in the polynomial representation. The elements $\pi_i$ satisfy the commutation relation 
$$
[\pi_i, \pi_j]  = m(\pi_i - \pi_j) s_{ij}
$$ 
for any $i, j$, which implies that the  elements $\sigma_k$ pairwise commute, cf. Polychronakos \cite{Poly92}.

Recall that operators ${\mathcal L}_k = {\rm Res }\, \sigma_k$
generate an algebra of (symmetric) quantum integrals of the trigonometric CM system. 
In particular ${\mathcal L}_1= E = \sum_{i=1}^N x_i \partial_i$ is the Euler operator, and ${\mathcal L}_2 = {\mathcal H}_N$, where ${\mathcal H}_N$ is the Hamiltonian of the trigonometric CM system  \eqref{CMtrig}.


Let us describe the action of the generating automorphisms $a_{\tau, l}$, $b_{\tau, l} \in \mathbb G$ in terms of the adjoint action of the elements
$$
y^l := \sum_{i=1}^N y_i^l, \quad x^l := \sum_{i=1}^N x_i^l.
$$
 Namely, the following statement takes place (cf. \cite{EG02B}).

\begin{proposition}
\label{praut1}
The action of the automorphisms $a_{\tau, l}$ and $b_{\tau, l}$ can be given as  
\beq{autact2}
a_{\tau, l} (h) =  e^{ad_{\frac{\tau}{l+1} y^{l+1}}} h,  \quad 
b_{\tau, l} (h) =  e^{-ad_{\frac{\tau}{l+1} x^{l+1}}} h, \quad (h \in H_m),
\eeq
where $ad_A (X) = A X - X A$. 
\end{proposition}
\begin{proof}
It is well known that the exponential of any nilpotent derivation is a homomorphism, so
it is sufficient to prove formulas \eqref{autact2} for $h =x_i, y_i$, $(i=1,\ldots, N)$.
The following commutation relations can be checked directly:
\beq{powers}
[x_i, y^{l+1}] = - (l+1) y_i^{l }, \quad [y_i, x^{l+1}] =  (l+1) x_i^{l}.
\eeq
The right-hand side of the formula \eqref{autact2} for $a_{\tau, l}(x_i)$ has two non-zero terms only by formula \eqref{powers}, which are $a_{\tau, l}(x_i) = x_i +\tau y_i^l$ as required. It is also clear that the first formula \eqref{autact2} states $a_{\tau, l}(y_i) = y_i$ as required. The second formula \eqref{autact2} can be checked similarly.
\end{proof}
Recall the standard  formula from linear algebra
\beq{Adad}
Ad_{e^A} = e^{ad_A}, 
\eeq
where $Ad_C X = C X C^{-1}$.  

In the polynomial representation we have a well defined action of the operators of the form $e^{p(D)}$, $p\in \C[x]$ since $p(D)$ are locally nilpotent operators.  We also have the natural action of the operators $p(D)$ on the functions of the form $e^{-q(x)} r(x)$, where $r \in \C[x]$, $q \in \C[x]^{S_N}$, which allows us to consider the operator 
$$
Ad_{e^{q(x)}} p(D) := e^{q(x)} p(D)e^{-q(x)}
$$ 
acting on the space of polynomials  $\C[x]$.

In view of formula \eqref{Adad} we can now restate  Proposition \ref{praut1} as follows. 

\begin{proposition}
\label{conjaut}
In the polynomial representation the action of elements $a_{\tau, l}(h)$ and $b_{\tau, l}(h)$  for any $h \in H_m$ can be given as
\beq{autact}\nonumber
a_{\tau, l} (h) =  
e^{\frac{\tau}{l+1} D^{l+1}} h e^{-\frac{\tau}{l+1} D^{l+1}} , 
\quad
b_{\tau, l} (h) = 
e^{-\frac{\tau}{l+1} x^{l+1}} h e^{\frac{\tau}{l+1} x^{l+1}}, 
\eeq
where $D^{l+1} :=\sum_{i=1}^N D_i^{l+1}$.
 \end{proposition}

This leads us to another interpretation of the Lassalle--Nekrasov correspondence $\varphi_{LN}$ between quantum integrals  of the trigonometric and rational harmonic CM systems, sending  ${\mathcal L}_p$ to ${\mathcal L}_p^{\mathcal H}$, $p \in \Lambda_N$, see diagram \eqref{LN}. 

 Let ${\mathcal T}$  be the commutative subalgebra in $H_m$ generated by the elements 
\beq{pik}
\pi^k:= \sum_{i=1}^N \pi_i^k, \quad k \in \Z_{\ge 0}.
\eeq
Define  ${\mathcal R} = \gamma({\mathcal T})$ to be its image under the automorphism $\gamma = a_{-1,1} \in \mathbb G$. 
More explicitly, the algebra ${\mathcal R}$ is generated by the elements
$$
\gamma(\pi^k) = \sum_{i=1}^N (x_i y_i -y_i^2)^k.
$$
Let ${\mathcal D}_{CM}^T$ and ${\mathcal D}_{CM}^R$  be the algebras of quantum integrals of the trigonometric and rational harmonic CM systems generated by the differential operators ${\mathcal L}_p$ and  ${\mathcal L}_p^{\mathcal H}$, $p \in \Lambda_N$, respectively.

\begin{proposition}
The following diagram is commutative:
\begin{equation}\label{LNv2}\nonumber
\begin{tikzcd}
{\mathcal T}  \arrow[r,"\gamma"] \arrow[d,swap,"\rm{Res}"] &
{\mathcal R} \arrow[d,"\rm{Res}"] \\
{{\mathcal D}_{CM}^T} \arrow[r,"\varphi_{LN}"] & {{\mathcal D}_{CM}^R},
\end{tikzcd}
\end{equation}
where $\varphi_{LN}$ is the Lassalle--Nekrasov correspondence between the quantum integrals: 
\begin{equation}
\label{HamCS}\nonumber
\varphi_{LN}({\mathcal L}_p) =  {\mathcal L}^{\mathcal H}_p,
\end{equation}
with the operator  ${\mathcal L}^{\mathcal H}_p$ defined by \eqref{cLpH}. Furthermore, this operator can also be given as
\beq{HamCMSn} 
 {\mathcal L}^{\mathcal H}_p   = e^{-L/2} {\mathcal L}_p e^{L/2}, \quad 
L = \Delta - \sum_{i<j}^N\frac{2m}{x_i - x_j}(\partial_i - \partial_j). 
\eeq
\end{proposition}
 \begin{proof}
All the maps in the diagram are homomorphisms so it is sufficient to compare the maps $\varphi_{LN}\circ {\rm Res}$ and ${\rm Res}\circ \gamma$ applied to $\pi^k\in \mathcal T$.

Recall that $L$ given in \eqref{HamCMSn} is the Hamiltonian of the rational CM system (without harmonic term) which can be expressed  as
$$
L = {\rm Res} \sum_{i=1}^N y_i^2,
$$
see  \cite{Dun89, Hec91}. 
Therefore by Proposition \ref{conjaut}  we have ${\rm Res}\, (\gamma(\pi^k))= e^{-L/2} {\mathcal L}_k e^{L/2}$, where
${\mathcal L}_k = {\rm Res}\, \pi^k$ is the quantum integral of the trigonometric CM system. 

 On the other hand 
by Proposition \ref{heat} and Theorem \ref{intertwine} we have 
$e^{-L/2} {\mathcal L}_k e^{L/2} = {\mathcal L}^{\mathcal H}_k$ as required.
\end{proof}

\section{Higher analogues of multivariable $m$-Hermite polynomials and bispectrality}
\label{sechigher}

Consider now the triangular automorphism  $\gamma$
of the rational Cherednik algebra $H_m$  which is a composition of automorphisms \eqref{automA}:
 \beq{automAgen}
\gamma(x_i) = x_i + {\mathcal P}_\gamma (y_i), \quad \gamma(y_i) = y_i, \quad \gamma(w)=w,
\eeq
where ${\mathcal P}_\gamma (z) = \sum_{j=0}^l \tau_j z^j$ is a polynomial of degree $l\in \Z_{\ge 0}$, $\tau_j \in \C$. 
The same arguments as above in the case $l=1$ lead to the next proposition  (see also \cite{Hor16}, \cite{Hor18}, \cite{VZh09} who used automorphisms of the Weyl algebra in a related context).

Recall that  $L_j={\rm Res} \sum_{i=1}^N D_i^j$ defines a higher order quantum integral of the rational CM system, see  Heckman  \cite{Hec91A}. 
 Define $L_\gamma$ as 
$$
L_\gamma = \sum_{j=0}^l \tau_j L_{j+1}.
$$
\begin{proposition}
\label{prophigher}
The algebra ${{\mathcal D}_{CM}^{R, \gamma}}  = {\rm Res} \, {\mathcal R}_\gamma$, where ${\mathcal R}_\gamma= \gamma({\mathcal T})$, is a 
  commutative algebra  of differenial operators with rational coefficients which is isomorphic to the algebra of polynomials in $N$ variables. The lowest order differential operator in the algebra ${{\mathcal D}_{CM}^{R, \gamma}}$ has the form
\beq{Kgamma}
K_\gamma ={\rm Res} \, \gamma (\sum_{i=1}^N x_i y_i) = {\rm Res} \sum_{i=1}^N (x_i y_i +  y_i {\mathcal P}_\gamma(y_i)) = E + { L}_{\gamma},
\eeq
where $E=\sum x_i \partial_i$ is the Euler operator.
\end{proposition}
Note that this proposition holds for any value of $m$, not necessarily integer. 
For integer $m$ the operator $K_\gamma$ admits an operator intertwining it with a corresponding constant coefficient differential operator.
Let $T=T(x, \partial)$ be a differential operator such that the BA function is given as 
\beq{Top}
\phi(x, \lambda) = T e^{(x,\lambda)}.
\eeq
We have the following result,  which for $l=1$ is contained in  \cite{BC20}.

\begin{proposition}
For any $m\in \Z_{\ge 0}$ the operator $K_\gamma$ satisfies the intertwining relation
$$
K_\gamma T = T (\sum_{j=0}^l \tau_j \partial^{j+1} + E),
$$
where $\partial^{j} = \sum_{i=1}^N \partial_i^{j}$.
\end{proposition}
\begin{proof}
Relation \eqref{Top} and the property $L_{j} \phi(x, \lambda) = \lambda^{j} \phi(x, \lambda)$, where $\lambda^j = \sum_{i=1}^N \lambda_i^j$,  imply the intertwining relation
$$
L_{j} T = T \partial^{j}.
$$
The operator $T$ is homogeneous of degree 0 by the homogeneity relation \eqref{hom} of the BA function $\phi(x, \lambda)$. Hence $[E, T]=0$ and the statement follows.
\end{proof} 

Recall that the operator $L_{\gamma}$ preserves the space of quasi-invariant polynomials ${\mathcal Q}_m$ and it acts on this space as a locally nilpotent operator.  Similarly to Propositions~\ref{Cor:eigValEq} and \ref{heat}, we have the following result.
\begin{proposition}\label{Pqformula}
For any homogeneous $q \in {\mathcal Q}_m$ the polynomial 
\beq{Pq}
P^\gamma_q = e^{\sum_{j=0}^l \frac{\tau_j}{j+1}L_{j+1}} q
\eeq
is an eigenfunction of the operator $K_\gamma$:
$$
K_\gamma P^\gamma_q =\deg q \,  P^\gamma_q.
$$
More generally, we have the operator $e^{\sum_{j=0}^l \frac{\tau_j}{j+1}L_{j+1}}$ intertwining the action of the algebras ${{\mathcal D}_{CM}^{R, \gamma}}$ and ${{\mathcal D}_{CM}^{T}}$ on quasi-invariants ${\mathcal Q}_m$.
\end{proposition}
The spectrum of the operator $K_\gamma$ on ${\mathcal Q}_m$ thus coincides with the spectrum of the Euler operator $E$.

In the simplest one-dimensional case with  $m=0$  and $\tau_j=0$ for $j<l$, $\tau_l = \tau$ we have the operator
$$
K_\gamma = x \frac{d }{d x} + \tau \frac{d^{l+1}}{d x^{l+1}}.
$$
The corresponding polynomials 
$$
\displaystyle P_n(x; \tau) = e^{\frac{\tau}{l+1}  \frac{d^{l+1}}{d x^{l+1}}} x^{n}
$$
are known as Gould--Hopper polynomials \cite{GH62}. 
These polynomials appear in relation with various problems from probability, quantum mechanics and integrable systems \cite{Cha11, DLMT96, VL13}.

Note that for $l\ge 2$ these polynomials do not satisfy any three-term recurrence relation and hence by Favard's theorem do not form an orthogonal system (but  satisfy a more general vector orthogonality property, see Horozov \cite{Hor16}).

The corresponding eigenvalue equation with $ \, \tau=-l-1$
\beq{freud}
K_\gamma \Psi= x \frac{d\Psi }{d x} - (l+1) \frac{d^{l+1}\Psi}{d x^{l+1}}=\lambda \psi 
\eeq
has interesting non-polynomial solutions 
$$
\Psi=m_k=\int_0^{+\infty}z^{k}e^{-z^{2(l+1)}+xz^2}dz, \quad k\ge 0, 
$$
which are the moments of the generalised Shohat--Freud weight \cite{Fre76} (see e.g. \cite{CJ20} where moments are expressed through generalised hypergeometric functions for $l=2,3,4$). Indeed,
it is easy to see that 
$$
K_\gamma e^{-z^{2(l+1)}+xz^2}=\frac{1}{2}z\frac{d}{dz} e^{-z^{2(l+1)}+xz^2},
$$
which after integration by parts gives the relation (\ref{freud}) with $\lambda=-\frac{k}{2}.$

Going back to the general case consider the function
$$
F_\gamma (x, \lambda) = \phi(x, \lambda) e^{\sum_{j=0}^l  \frac{\tau_j}{j+1} \lambda^{j+1}},
$$
where $\lambda^{j+1}= \sum_{i=1}^N \lambda_i^{j+1}$. For a homogeneous basis $q^i$, $i \in \Z_{\ge 0}$  in the space of quasi-invariants ${\mathcal Q}_m$  consider the polynomials 
$P^\gamma_i(x):= P^\gamma_{q_i}(x)$ given by formula \eqref{Pq}, where $q_i$ is the dual basis in quasi-invariants ${\mathcal Q}_m$. 

\begin{proposition}
\label{Flexpansion}
We have the following expansion
$$
 F_\gamma(x, \lambda) = \sum_i P^\gamma_i(x) q^i(\lambda).
$$
\end{proposition}
\begin{proof}
The function $F_\gamma(x, \lambda) $ can be expanded as  
$$
 F_\gamma(x, \lambda) = \sum_i \widetilde P^\gamma_i(x) q^i(\lambda)
$$
for some quasi-invariant polynomials $\widetilde P^\gamma_i(x)$ with the highest degree term $q_i(x)$. Similarly to the proof of Proposition \ref{Cor:eigValEq} we have the relation
$$
K_\gamma  F_\gamma(x, \lambda) = E_\lambda F_\gamma(x, \lambda),
$$
where the operator $K_\gamma$ is given by formula \eqref{Kgamma}. Therefore
\beq{difeqK}
K_\gamma \widetilde P^\gamma_i(x) = \deg q^i \widetilde P^\gamma_i(x).
\eeq
Since the polynomial solution of differential equation \eqref{difeqK} with the highest degree term $q_i$ is unique, it follows by Proposition \ref{Pqformula} that  $\widetilde P^\gamma_i(x) =  P^\gamma_i(x)$. 
\end{proof}

This function satisfies the following bispectrality property in the sense of Duistermaat and Gr\"unbaum \cite{DG86}.
\begin{proposition}
\label{prop75}
For any quasi-invariant polynomial $q\in {\mathcal Q}_m$ the following differential equations hold:
\beq{bisp0}\nonumber
L_{q, x} F_\gamma(x, \lambda) = q(\lambda) F_\gamma(x, \lambda), \quad 
\widehat L_{q, \lambda}  F_\gamma(x, \lambda) = q(x) F_\gamma(x,  \lambda),
\eeq
where $\widehat L_{q, \lambda} = e^{\sum_{j=0}^l  \frac{\tau_j}{j+1} \lambda^{j+1}} L_{q, \lambda} e^{-\sum_{j=0}^l  \frac{\tau_j}{j+1} \lambda^{j+1}}$. 
\end{proposition}

The proof follows directly from  property \eqref{BAeigenf} of the BA function and its symmetry \eqref{sym}. This leads to the following result.
\begin{theorem}
\label{ddb}
For any homogeneous quasi-invariant $q(x) \in {\mathcal Q}_m[d]$ the polynomials $P^\gamma_{i}(x)$ satisfy the recurrence relation  
\beq{bisp}
q(x) P^\gamma_{i}(x) =  \sum_ {j:   d_i -l d\le  \deg q_j \le d_i +d} C^j_{i} P^\gamma_{j}(x),  \quad C^j_i = \langle \widehat L_q q^j, q_i \rangle_m,
\eeq
where $d_i = \deg P_i^\gamma = \deg q_i$, 
which together with the differential equation 
$$
K_\gamma P^\gamma_{i}(x) =d_i \,  P^\gamma_{i}(x)
$$
 gives the differential--difference bispectral property for the polynomials $P^\gamma_{i}(x)$.
\end{theorem}
 \begin{proof}
Let us write the operator $\widehat L_{q, \lambda}$ in the basis $q^i(\lambda)$:
\beq{Lqlam}
\widehat L_{q, \lambda}  q^i(\lambda) = \sum_{j}  C^i_j q^j(\lambda),
\eeq
where the coefficients $C^i_j = \langle \widehat L_q q^i, q_j \rangle_m$ since $\langle q^j, q_k\rangle_m = \delta^j_k$. The operator $\widehat L_{q, \lambda}$ can be given by 
replacing partial derivatives $\partial_i$ in $L_{q, \lambda}$ with $\partial_i - {\mathcal P}_\gamma(\lambda_i)$. Since  the operator $L_{q, \lambda}$ is lowering the degree of quasi-invariants by $d$ and the multiplication by ${\mathcal P}_\gamma(\lambda_i)$ is increasing the degree by $l$ the inner product  $C^i_j = \langle \widehat L_q q^i, q_j \rangle_m=0$  if $\deg q_j>d_i +ld$ or $\deg q_j < d_i -d$.

By applying the operator $\widehat L_{q, \lambda}$ to the function $F_\gamma(x, \lambda)$ and using 
  formula \eqref{Lqlam}  and Propositions \ref{Flexpansion} and \ref{prop75}  we get the required relation  \eqref{bisp}.
\end{proof}

Let us illustrate Theorem \ref{ddb} by a one-dimensional example with ${\mathcal P}_\gamma (z) = \tau z^l$. Let us define polynomials $p_k^l$, $k\in \Z_{\ge0}$ by the expansion
$$
\phi_m(x, \lambda)e^{\frac{\tau}{l+1}\lambda^{l+1}} = \sum_{k=0}^\infty p_k^l(x) \lambda^k,
$$
where the BA function $\phi_m$ is given by \eqref{ba2}. 
Note that $p_{2s-1}^l(x)=0$ for $s=1, \ldots, m$. 
The relation of polynomials $p^l_k$ at $l=1$, $\tau=-1$ with $m$-Hermite polynomials is given by the formula
$$
H_k^{(m)}(x) = k! p^1_k(x).
$$

In this case the operator 
$
\widehat L_{q, \lambda}
$
for $q=x^2$ takes the form
$$
\widehat L_{x^2, \lambda} = \partial^2_\lambda  - 2(\tau \lambda^l  + m \lambda^{-1})\partial_\lambda   +  \tau^2 \lambda^{2l} +(2m- l)\tau \lambda^{l-1},
$$
and hence
\beq{klkl0}
x^2 p^l_k = (k+2)(k-2m+1)p^l_{k+2} +   \tau(l - 2(k-m+1))p^l_{k-l+1}   + \tau^2 p^l_{k-2l}.
\eeq

Let now   $m=1$. Then $x^3 \in {\mathcal Q}_m$ and the corresponding differential operators $L_{x^3, \lambda}$ and $\widehat L_{x^3, \lambda}$ have the form 
$$
L_{x^3, \lambda} = \partial_\lambda^3 - 3 \lambda^{-1} \partial_\lambda^2+ 3 \lambda^{-2}\partial_\lambda,
$$
\begin{multline*}
\widehat L_{x^3, \lambda} = L_{x^3, \lambda}  - 
 3 \tau \lambda^l\partial_\lambda^2 
+3(\tau(2-l)\lambda^{l-1}  + \tau^2 \lambda^{2 l} )\partial_\lambda
\\
- \tau(l-1)(l-3) \lambda^{l-2} + 3(l-1) \tau^2 \lambda^{2l-1} - \tau^3 \lambda^{3l}.
\end{multline*}
The corresponding recurrence relation has the form
\begin{multline}
\label{klkl}
x^3 p^l_k = (k+3)(k+1)(k-1) p^l_{k+3}  
\\
- \tau (   3 k^2  - 3 k l + l^2 + 3 k- l -3    ) p^l_{k-l+2} + 3 \tau^2(k-l)p^l_{k-2l+1}
-\tau^3 p^l_{k-3l}.
\end{multline}

We note that when $l=1$ recurrence relations of this type were considered by G\'omez-Ullate et al. \cite{GKKM16}  
in the context of exceptional Hermite polynomials \cite{GGM14}. In this case formulas \eqref{klkl0}, \eqref{klkl} take the form of the following recurrence relations for $m$-Hermite polynomials:
\beq{mHrec1}\nonumber
x^2 H_k^{(1)} = \frac{k-1}{k+1} H_{k+2}^{(1)} + (2k-1) H_{k}^{(1)} + k(k-1) H_{k-2}^{(1)},
\eeq
\beq{mHrec2}\nonumber
x^3 H_k^{(1)} = \frac{k-1}{k+2} H_{k+3}^{(1)} + 3(k-1) H_{k+1}^{(1)} + 3k(k-1)H_{k-1}^{(1)}+ k(k-1)(k-2)H_{k-3}^{(1)}.
\eeq


\section{Deformed Calogero--Moser systems}
\label{secDefCM}

Let us recall the following generalised CM Hamiltonian which describes interaction of two types of particles \cite{CFV98}, \cite{SV04}:
\begin{multline}
\label{defHrat}
 L_{N_1, N_2}  = k^{-1}\sum_{i=1}^{N_1} \partial^2_{z_i}  + \sum_{j=1}^{N_2} \partial^2_{w_j}  - \sum_{i_1<i_2}^{N_1}\frac{2}{z_{i_1} - z_{i_2}}(\partial_{z_{i_1}}-\partial_{z_{i_2}}) 
\\
- 
 \sum_{j_1<j_2}^{N_2}\frac{2 k^{-1}}{w_{j_1} - w_{j_2}}(\partial_{w_{j_1}}-\partial_{w_{j_2}})
- \sum_{i=1}^{N_1}\sum_{j=1}^{N_2} \frac{2}{z_i - w_j}(k^{-1}\partial_{z_i} - \partial_{w_j}),
\end{multline}
where $N_1, N_2 \in \Z_{\ge 0}$. This Hamiltonian is known to be integrable for any non-zero $k\in \C$.

Consider the rational Cherednik algebra $H_{1/k}$ associated with the symmetric group $S_N$, where $k\in \N$, $k\le N$. In this case  the polynomial representation $\C[x]$ has a submodule $I$ consisisting of polynomials $p(x)$ vanishing on the $S_N$ orbit of the subspace $\Pi$ given by 
\begin{equation*}
\begin{split}
x_1 &=\ldots = x_k =: z_1,\\
x_{k+1} &=\ldots = x_{2k} =: z_2,\\
&\ \, \vdots\\
x_{(N_1-1) k +1} &= \ldots = x_{N_1 k} =: z_{N_1},
\end{split}
\end{equation*}
where $N_1\in \N$ is such that $N_2:= N- k N_1 \in \N$ (\cite{Fei12},  \cite{ESG09}).

Let us consider the action of an $S_N$--invariant element $h \in H_{1/k}$ in the quotient module $\C[x]/I$. It gives rise to a differential operator ${\rm Res}_\Pi \, h$ acting on functions in variables $z, w$, where $z=(z_1, \ldots, z_{N_1})$, and $w=(w_1, \ldots, w_{N_2})$ is given by $w_i = x_{N_1 k +i}$, $i=1, \ldots, N_2$.  
Let $\Sigma\subset \Pi$ be the union of hyperplanes in $\Pi$ where two coordinates from the set $z_i, w_j$ become equal.
Then this operator may be defined by the formula
$$
{\rm Res}_\Pi h (f)  = h(\widetilde f)|_\Pi,
$$
where $f$ is a regular function on a small open set $U \subset\Pi\setminus \Sigma$, and  $\widetilde f$ is its $S_N$-invariant regular extension to an $S_N$-invariant open set $V\subset \C^N$, $U \subset V$, such that $\widetilde f|_U = f$.
It was established in \cite{Fei12} that 
\beq{Fres}
{\rm Res}_\Pi \, \sum_{i=1}^N y_i^2 =  L_{N_1, N_2}.
\eeq

Let us also consider the trigonometric version of the Hamiltonian \eqref{defHrat} which (in the exponential coordinates) has the form \cite{CFV98}, \cite{SV04}
\begin{multline}
\label{defHtrig}\nonumber
{\mathcal L}_{N_1, N_2} =
k^{-1}\sum_{i=1}^{N_1} (z_i \partial_{z_i})^2 +  \sum_{j=1}^{N_2} (w_j \partial_{w_j})^2 -
 \sum_{i_1<i_2}^{N_1} \frac{z_{i_1}+z_{i_2}}{z_{i_1}-z_{i_2}}(z_{i_1} \partial_{z_{i_1}} - z_{i_2} \partial_{z_{i_2}})
\\
- k^{-1} \sum_{j_1<j_2}^{N_2} \frac{w_{j_1}+w_{j_2}}{w_{j_1}-w_{j_2}}(w_{j_1} \partial_{w_{j_1}} - w_{j_2} \partial_{w_{j_2}})
- \sum_{i=1}^{N_1}\sum_{j=1}^{N_2} \frac{z_i+ w_j}{z_i - w_j} (k^{-1} z_i \partial_{z_i} -  w_j \partial_{w_j}).
\end{multline}
\begin{lemma}
\label{EulerNN}
The following relations hold:
\beq{deftr}
{\rm Res}_\Pi \sum_{i=1}^N (x_i y_i)^2 = {\mathcal L}_{N_1, N_2},
\eeq
and
\beq{Eulerd}
{\rm Res}_\Pi \sum_{i=1}^N x_i y_i =  \sum_{i=1}^{N_1} z_i \partial_{z_i} + \sum_{j=1}^{N_2} w_j \partial_{w_j} =: E_{N_1, N_2}.
\eeq
\end{lemma}
\begin{proof}
Recall that 
$$
{\mathcal H}_N = {\rm Res} \sum_{i=1}^N {\mathcal D}_i^2 =  {\rm Res} \sum_{i=1}^N (x_i y_i)^2  =
\sum_{i=1}^N (x_i \partial_{x_i})^2 - k^{-1} \sum_{i<j}^N \frac{x_i+x_j}{x_i-x_j} (x_i \partial_{x_i} - x_j \partial_{x_j}). 
$$
Let us consider orthonormal coordinates $(\tilde z, w)$ on the plane $\Pi$, where $\tilde z = k^{1/2} z$. In order to calculate the operator ${\rm Res}_\Pi \sum_{i=1}^N (x_i y_i)^2$ one has to replace in the operator ${\mathcal H}_N$ coordinates $x_{j k +s}$ with $k^{-1/2}\tilde z_j$  ($0\le j \le N_1-1$, $1\le s \le k$), and derivatives  $\partial_{x_{j k +s}}$ with $k^{-1/2} \partial_{\tilde z_j}$, coordinates $x_{N_1 k +i}$ with $w_i$, $i=1, \ldots, N_2$ and derivatives $\partial_{x_{N_1 k +i}}$ with $\partial_{w_i}$ and to replace with zero all the derivatives normal to the plane $\Pi$ (cf. \cite{Fei12}). This produces the operator
\begin{multline*}
{\rm Res}_\Pi \sum_{i=1}^N (x_i y_i)^2 = k^{-1}\sum_{i=1}^{N_1} 
 ({\tilde z}_i \partial_{{\tilde z}_i})^2 +  \sum_{j=1}^{N_2} (w_j \partial_{w_j})^2 -
 \sum_{i_1<i_2}^{N_1} \frac{{\tilde z}_{i_1}+ {\tilde z}_{i_2}}{{\tilde z}_{i_1}- {\tilde z}_{i_2}}({\tilde z}_{i_1} \partial_{{\tilde z}_{i_1}} - {\tilde z}_{i_2} \partial_{{\tilde z}_{i_2}})
\\
- k^{-1} \sum_{j_1<j_2}^{N_2} \frac{w_{j_1}+w_{j_2}}{w_{j_1}-w_{j_2}}(w_{j_1} \partial_{w_{j_1}} - w_{j_2} \partial_{w_{j_2}})
- \sum_{i=1}^{N_1}\sum_{j=1}^{N_2} \frac{k^{-1/2}\tilde z_i+ w_j}{k^{-1/2}\tilde z_i - w_j} (k^{-1} \tilde z_i \partial_{{\tilde z}_i} - w_j \partial_{w_j}).
\end{multline*}
Formula \eqref{deftr} follows by changing variables $\tilde z= k^{1/2} z$. Formula \eqref{Eulerd} is easy to see in a similar way.
\end{proof}

Define the algebra $\Lambda_{N_1,N_2}\subset \C[z, w]$ as the algebra generated by the deformed Newton sums (cf. \cite{SV04})
$$
p_r = k \sum_{i=1}^{N_1} z_i^r + \sum_{i=1}^{N_2} w_i^r, \quad r\ge 0.
$$
Let $I^{S_N}$ be the $S_N$-invariant part of the ideal $I$, it consists of $S_N$-invariant polynomials $p(x)$ which vanish on $\Pi$.
\begin{lemma}
\label{LambdaQuot}
We have an isomorphism of algebras $\Lambda_{N_1,N_2}\cong \Lambda_N/I^{S_N}$.
\end{lemma}
\begin{proof}
Consider the map $\varphi: \Lambda_N \to \Lambda_{N_1, N_2}$ given by the natural restriction $\varphi(p)=p|_\Pi$, $p\in \Lambda_N$. It is clear that ${\rm Im } \,  \varphi = \Lambda_{N_1, N_2}$ and $\Ker \,  \varphi = I^{S_N}$.
\end{proof}

We have the following version of the Lassalle--Nekrasov correspondence for the deformed CM operators, (which also can be extracted from \cite{DH12}).
\begin{theorem}
\label{LNdeform}
The following diagram is commutative for any $p\in \Lambda_N$:
\begin{equation}\label{LNdef5}\nonumber
\begin{tikzcd}
\Lambda_{N_1, N_2} \arrow[r,"\chi_{N_1, N_2}"] \arrow[d,swap,"{{\mathcal L}}_{p, N_1, N_2}"] &
\Lambda_{N_1, N_2} \arrow[d,"{\mathcal L}^{\mathcal H}_{p, N_1, N_2}"] \\
\Lambda_{N_1, N_2} \arrow[r,"\chi_{N_1, N_2}"] & \Lambda_{N_1, N_2},
\end{tikzcd}
\end{equation}
with $\chi_{N_1, N_2}(f) = e^{-\frac12 L_{N_1, N_2}} f$, $f\in \Lambda_{N_1, N_2}$, and ${\mathcal L}_{p, N_1, N_2}= {\rm Res}_\Pi \, p(\pi_1, \ldots, \pi_N)$, 
${\mathcal L}^{\mathcal H}_{p, N_1, N_2} = Ad_{e^{-\frac12 L_{N_1, N_2}}} {\mathcal L}_{p, N_1, N_2}$.

Furthermore, the following equalities hold: 
\beq{pp11}
{\mathcal L}_{\sum x_i y_i, N_1, N_2} = E_{N_1, N_2},  \quad {\mathcal L}_{\sum (x_i y_i)^2, N_1, N_2} = {\mathcal L}_{N_1, N_2},
\eeq
and
\beq{pp22}
{\mathcal L}^{\mathcal H}_{\sum x_i y_i, N_1, N_2} = -  L_{N_1, N_2} + E_{N_1, N_2}.
\eeq
\end{theorem}
\begin{proof}
Let us recall commutative diagram \eqref{LN}. This diagram is equivalent to the following commutative diagram where Cherednik algebra elements act via Dunkl operators:
\begin{equation}\label{LNdef6}\nonumber
\begin{tikzcd}
\Lambda_N \arrow[r,"\chi"] \arrow[d,swap,"p(\pi)"] &
\Lambda_N \arrow[d,"p^{\mathcal H}(\pi)"] \\
\Lambda_N \arrow[r,"\chi"] & \Lambda_N 
\end{tikzcd}
\end{equation}
with $\chi(f) = {e^{-\frac12 \sum y_i^2}}(f)$, $f \in \Lambda_N$, and the operator $p^{\mathcal H}(\pi) = Ad_{e^{-\frac12 \sum y_i^2}} p(\pi)$. By considering actions in the quotient space $\Lambda_{N_1, N_2}$ (see Lemma \ref{LambdaQuot}) and applying the operation ${\rm Res}_\Pi$ and formula \eqref{Fres} the first statement follows.

Formulas \eqref{pp11} follow by Lemma \ref{EulerNN}. In order to establish \eqref{pp22} let us observe that $p^{\mathcal H}(\pi) = {\rm Res}_\Pi (x_i y_i - \sum y_i^2)$. Then formula \eqref{pp22} follows from \eqref{Fres}.
\end{proof}

We also have a quasi-invariant version of Theorem \ref{LNdeform}. Let us introduce the following algebra $\Lambda_{N_1, N_2}(k)\subset \C[z_1, \ldots, z_{N_1}, w_1, \ldots, w_{N_2}]$ of generalised quasi-invariants following \cite{ERF16}. A polynomial $q\in \Lambda_{N_1, N_2}(k)$ if the following three conditions are satisfied:
\begin{enumerate}
\item
$\sigma(q)=q$ for any $\sigma \in S_{N_2}$ which acts by permuting $w$-variables,
\item
for any fixed $w$, the polynomial $q$ is $k$-quasi-invariant as a polynomial in variables $z_1, \ldots, z_{N_1}$,
\item
$\partial_{z_i} q = k \partial_{w_j}q$ if $z_i=w_j$ for any $1\le i \le N_1$, $1 \le j \le N_2$.
\end{enumerate}
Note that we have an inclusion of algebras $\Lambda_{N_1, N_2} \subset \Lambda_{N_1, N_2}(k)$. We also have the property that the action of the differential operator $L_{N_1, N_2}$ preserves the space  $\Lambda_{N_1, N_2}(k)$, and, furthermore,  $\Lambda_{N_1, N_2}(k)$ is a module over the spherical subalgebra $eH_{1/k}e$ \cite{ERF16}. In this module the elements $p(\pi_1, \ldots, \pi_N)$, $p\in \Lambda_N$ act as differential operators ${\mathcal L}_{p, N_1, N_2}$, in particular, these operators preserve the space $\Lambda_{N_1, N_2}(k)$. This leads us to the following extension of Theorem \ref{LNdeform}.
\begin{theorem}
\label{LNdeformQuasi}
The following diagram is commutative for any $p\in \Lambda_N$:
\begin{equation}\label{LNdef7}\nonumber
\begin{tikzcd}
\Lambda_{N_1, N_2}(k) \arrow[r,"\chi_{N_1, N_2}"] \arrow[d,swap,"{{\mathcal L}}_{p, N_1, N_2}"] &
\Lambda_{N_1, N_2}(k) \arrow[d,"{\mathcal L}^{\mathcal H}_{p, N_1, N_2}"] \\
\Lambda_{N_1, N_2}(k) \arrow[r,"\chi_{N_1, N_2}"] & \Lambda_{N_1, N_2}(k).
\end{tikzcd}
\end{equation}
\end{theorem}
Indeed, as we explained all the maps from the space $\Lambda_{N_1, N_2}(k)$ to itself are well-defined. Since a differential operator acting on functions in variables $z, w$ is fully determined by its action on $\Lambda_{N_1, N_2}$ the statement follows from Theorem \ref{LNdeform}.

Operators ${\mathcal L}_{p, N_1, N_2}$ in Theorems \ref{LNdeform}, \ref{LNdeformQuasi} were obtained as restrictions ${\rm Res}_\Pi$ of the elements $Ad_{e^{-\frac12 \sum y_i^2}} p(\pi)$ of the rational Cherednik algebra. In terms of automorphisms of the rational Cherednik algebras considered in Section \ref{automLN} we have ${\mathcal L}^{\mathcal H}_{p, N_1, N_2}  = {\rm Res}_\Pi a_{-1,1}(p(\pi))$, where $a_{-1,1}$ is now an automorphism of the algebra $H_{1/k}$ defined in \eqref{automA}.

Let us now consider a more general automorphism  $\gamma$ of the algebra $H_{1/k}$ given by formula \eqref{automAgen}. Define the algebra of differential operators ${{\mathcal D}_{N_1, N_2}^{R, \gamma}}  = {\rm Res}_\Pi \, {\mathcal R}_\gamma$, where ${\mathcal R}_\gamma= \gamma({\mathcal T})$, and ${\mathcal T}$ is the commutative subalgebra in the algebra $H_{1/k}$ generated by elements $\pi^r$, $r\in \Z_{\ge 0}$,  given by formula \eqref{pik}.

We have the following generalisation of Proposition \ref{prophigher} to the deformed case.
\begin{theorem}
\label{prophigherdeform}
The algebra ${\mathcal D}_{N_1, N_2}^{R, \gamma}$  is a 
  commutative algebra  of differenial operators with rational coefficients which contains $N_1+N_2$ algebraically independent elements. The lowest order differential operator in the algebra ${{\mathcal D}_{N_1, N_2}^{R, \gamma}}$ has the form
\beq{KgammaDeformed}
K_{N_1, N_2, \gamma} ={\rm Res}_\Pi \, \gamma (\sum_{i=1}^N x_i y_i) = {\rm Res}_\Pi \sum_{i=1}^N (x_i y_i +  y_i {\mathcal P}_\gamma(y_i)) = E_{N_1, N_2} + { L}_{N_1, N_2, \gamma},
\eeq
where  
$$
{L}_{N_1, N_2, \gamma} = {\rm Res}_\Pi \sum_{i=1}^N y_i {\mathcal P}_\gamma (y_i)
$$ 
is the order $l+1$ quantum integral of the rational generalised CM Hamiltonian \eqref{defHrat}.
\end{theorem}
Indeed, algebraic independence of (the highest order terms of) the operators 
$$
{\rm Res}_\Pi \gamma(\pi^s)= {\rm Res}_\Pi \sum_{i=1}^N (x_i y_i +y_i {\mathcal P}_\gamma (y_i))^s
$$
 with $s= 1, \ldots, N_1+N_2$ follows from \cite[Proposition 4]{SV04}. The rest of the proof is immediate from Lemma \ref{EulerNN} and other previous considerations.

In the case $P_\gamma(z)=-z$ the operator \eqref{KgammaDeformed} is equal to the operator $-  L_{N_1, N_2} + E_{N_1, N_2}={\mathcal L}^{\mathcal H}_{\sum x_i y_i, N_1, N_2}$ considered above.
Its integrability was established independently in \cite{Fei12} for integer $k$ and in \cite{DH12} for general $k\in\C$. Furthermore, in the latter paper so-called super Hermite polynomials were introduced and shown to be joint eigenfunctions of ${\mathcal L}^{\mathcal H}_{\sum x_i y_i, N_1, N_2}$ and its higher order integrals.


\section{Concluding remarks}

We have defined a new class of $\mathcal A$-Hermite polynomials related to special configurations $\mathcal A$ of hyperplanes with multiplicities, and we showed that these polynomials are the eigenvectors of the corresponding generalised CM operator (\ref{harm}). Of particular interest is the case of multivariable $m$-Hermite polynomials. A natural question is whether these polynomials are joint eigenvectors of the corresponding commuting quantum integrals. The answer is in general negative already for the symmetric $m$-Hermite polynomials.

The reason is that if the multiplicity $m$ is non-zero then the bilinear forms on the space $\cQ_m$ become indefinite and the invariant subspaces of the algebra of integrals in general are not one-dimensional and contain multivariable $m$-Hermite polynomials as generalised eigenvectors. 

One can see this already in the simplest two-particle case, when the second integral has Jordan blocks. Indeed, let us consider the trigonometric side of the Lassalle--Nekrasov correspondence instead. Then this integral corresponds to the second order Hamiltonian \eqref{CMtrig}. Let us fix $m=1$. 
Then for any $l\in \N$ we have a $2\times 2$ Jordan block given by
\begin{equation*}
{\mathcal H}_2 (x_1^l x_2^l) = 2 l^2 x_1^l x_2^l, \quad {\mathcal H}_2 (x_1^{l-1} x_2^{l-1}(x_1^2+x_2^2))  = 2 l^2 x_1^{l-1} x_2^{l-1}(x_1^2+x_2^2) - 4 x_1^l x_2^l.
\end{equation*}

The corresponding spectral decomposition of the space of quasi-invariants $\cQ_m$ seems to be very interesting.

Another interesting direction is to extend the considerations to the difference versions of the Calogero--Moser systems, see \cite{Rui87} and  \cite{vDie95}.
We hope to come back to some of these questions elsewhere.

\section{Acknowledgements}

We are very grateful to Nikita Nekrasov for attracting our attention to his work \cite{Nek97}, 
and to Margit R\"osler for useful discussions. 

The second author (M.H.)~acknowledges financial support from the Swedish Research Council (Reg.~nr.~2018-04291).
The work of the first (M.F.) and third (A.P.V.) authors (Sections \ref{sec3}, \ref{automLN}, \ref{sechigher}) was supported by the Russian Science Foundation grant no.~20-11-20214.

\bibliographystyle{amsalpha}

\end{document}